\documentclass[11pt,a4paper,twoside,reqno]{amsart} 
\usepackage[utf8]{inputenc}
\usepackage[english]{babel}
\usepackage{graphicx}
\usepackage[width=160mm,top=25mm,bottom=25mm,bindingoffset=6mm]{geometry}
\usepackage{lmodern}

\usepackage{float}
\usepackage{xspace}
\usepackage{amssymb}
\usepackage{latexsym}
\usepackage{amsthm}
\usepackage{amsmath}
\usepackage{amssymb}
\usepackage{mathrsfs}
 \usepackage{multirow}
 \usepackage[normalem]{ulem}
\usepackage{color}

\usepackage[skip=2pt]{caption}
\usepackage{tensor}

\newtheorem{thm}{Theorem}[section]
\newtheorem{lem}[thm]{Lemma}
\newtheorem{cor}[thm]{Corollary}

\newtheorem{prop}[thm]{Proposition}
\theoremstyle{definition}
\newtheorem{defn}[thm]{Definition}
\newtheorem{oss}[thm]{Remark}
\newtheorem{example}[thm]{Example}

\DeclareMathOperator{\sgn}{sgn}

\newcommand{\bh}{\ensuremath{Q}\xspace}
\newcommand{\R}{\mathbb{R}}

\usepackage{hyperref}
\begin{document}
\title[Closed and broken electromagnetic orbits in Kerr--Newman spacetime]{Closed and broken electromagnetic orbits in Kerr--Newman spacetime}

\author[E. Caponio]{Erasmo Caponio}
\address{E. Caponio
\newline\indent Dipartimento di Meccanica, Matematica e Management\newline\indent
	Politecnico di Bari}
\email{erasmo.caponio@poliba.it}
    
\author[G. Sanzeni]{Giulio Sanzeni}
\address{G. Sanzeni
\newline\indent Faculty of Science - Mathematics $\&$ Statistical Sciences \newline\indent University of Alberta,  Edmonton, Canada}
\email{sanzeni@ualberta.ca}

\author[S. Suhr]{Stefan Suhr}
\address{S. Suhr
\newline\indent Fakult\"at f\"ur Mathematik \newline \indent Ruhr-Universit\"at Bochum,  Bochum,  Germany}
\email{stefan.suhr@ruhr-uni-bochum.de}

  \begin{abstract}
    		We study future-pointing timelike solutions of the Lorentz force equation in the
  		sub-extremal Kerr--Newman spacetime, with special attention to the time-machine
  		region $\mathfrak T$, where the axial Killing field $\partial_\phi$ is timelike.
  		We first construct smooth closed electromagnetic orbits tangent to $\partial_\phi$
  		in the positive equatorial part of $\mathfrak T$: the radius of such a circle
  		determines, and is determined by, the charge-to-mass ratio of the particle which must have opposite sign to that of the black hole charge.  We
  		then prove the existence of  spherical electromagnetic orbits contained in the
  		equatorial time-machine region and derive explicit relations between their radius,
  		charge-to-mass ratio, energy and angular momentum. Next we give sufficient
  		conditions ensuring that an equatorial electromagnetic orbit is a flyby
  		orbit with radial turning point in $\mathfrak T$. Finally, to describe
  		charged-particle decay processes whose fragments have different charge-to-mass
  		ratios, we introduce the notion of a broken electromagnetic orbit: a continuous, piecewise smooth, future-pointing 
  		worldline whose smooth pieces solve the Lorentz force equation. Imposing
  		conservation of kinetic four-momentum and electric charge at the decay vertices,
  		we exhibit an energy extraction  process  followed by a causality-violating one. The latter is
  		realized by a closed broken electromagnetic orbit, which we construct in both the
  		$r$-positive and the $r$-negative regions inside the inner horizon, by
  		concatenating two flyby branches sharing a common radial turning value $\bar r$, and having charge-to-mass ratios with opposite sign to that of the black hole charge,
  		with a spherical electromagnetic orbit of radius $\bar r$.
\end{abstract}

	\maketitle

	 \tableofcontents

\section{Introduction}
The existence of closed causal curves in rotating black hole spacetimes is a
classical source of interest in Lorentzian geometry and general relativity,
going back to the discovery of the Kerr \cite{Kerr-paper} and Kerr--Newman
\cite{NewmanEtAl1965} solutions and to Carter's analysis of their global
structure \cite{Carter_causality}. In every member of this family the
obstruction to global causality is localized in the \emph{time-machine region}
\[
\mathfrak{T}:=\{p\in\mathcal M : g_p(\partial_\phi,\partial_\phi)<0\},
\]
where the axial Killing field $\partial_\phi$ becomes timelike and its closed
integral circles are therefore closed timelike curves.  
In the Kerr spacetime the time-machine region is confined to the negative-$r$ region,  \cite{Calvani_timelike}. For a Kerr black
hole, i.e.\ when $a^2\le M^2$, $\mathfrak T$ lies below the Cauchy horizon
$r_-=M-\sqrt{M^2-a^2}$ (degenerating into the single
horizon $r=M$ at the extremal case),  and the causality violation it carries is invisible from the exterior
region $\{r>r_+\}$ \cite{Carter_causality,KBH_book}. In the over-extremal regime $a^2>M^2$, instead, no horizon is present \cite{Carter_causality,KBH_book}.  Consequently, the time-machine region, which accumulates at the ring singularity $\Sigma=\{r=0,\theta=\pi/2\}$, is no longer shielded by horizons \cite{Penrose_naked}.

It is then natural to ask whether this causality violation can be realized
along freely falling matter or light, that is, whether Kerr admits closed causal
geodesics.  In the
analytic extension of the  Kerr black hole across the
horizons and into the negative-$r$ region (the {\em Kerr-star} spacetime), there are
no closed timelike geodesics \cite{sanzeni_timelike}, while null geodesics can
be neither closed nor even contained in a compact subset \cite{sanzeni2024non}.
Thus, although a closed timelike curve passes through every point below the
inner horizon, the Kerr-star spacetime, contains no closed
causal geodesics. The timelike statement is the delicate one: once the
geodesics crossing the horizons, those confined to $\{0<r<r_-\}$ and those
tangent to the symmetry axis have been excluded, the remaining  case is that of spherical geodesics of constant negative radius with
negative Carter constant, which are ruled out by showing that the
Boyer--Lindquist time coordinate strictly increases over each full oscillation
of the polar angle $\theta$. The non-existence of closed null geodesics in the fast and extremal ($a^2\geq M^2$) Kerr  has  been proved in \cite{sanzmosani}.

The situation changes substantially in the Kerr--Newman spacetime, for two
independent reasons. First, as soon as the black-hole charge $\bh$ is non-zero,
the time-machine region intersects the positive-$r$ domain as well: on the
equatorial plane $g_{\phi\phi}(r,\pi/2)\approx-a^2\bh^2/r^2<0$ for $|r|$ small,
so $\mathfrak T$ surrounds the ring singularity on both sides (still inside the
Cauchy horizon, see Lemma~\ref{rtm}). Second, the natural worldlines of charged
matter are no longer geodesics but solutions of the Lorentz force equation
\begin{equation}\label{kF}
(\nabla_{\dot\gamma}\dot\gamma)^\alpha
= k\,F^\alpha{}_\beta\,\dot\gamma^\beta,
\end{equation}
where $k=e/m$ is the charge-to-mass ratio of the particle and
$F^\alpha{}_\beta=g^{\alpha\mu}F_{\mu\beta}$ is the $(1,1)$-tensor associated with the electromagnetic two-form $F$. The existence of worldline solutions to \eqref{kF} connecting chronologically related points in globally hyperbolic spacetimes was first established in \cite{Caponio_Minguzzi}, under the assumptions that $F$ is an exact $2$-form and $k$  belongs to a neighbourhood of zero. The result was then extended in \cite{Minguzzi_2003} to arbitrary $k$, provided that no null connecting geodesic exists between the chosen points.
The latter was hence generalized in \cite{Minguzzi_Sanchez} to causally related points for arbitrary $k$ in $C^2$ globally hyperbolic spacetimes.  However, since the Kerr--Newman spacetime fails to be causal below the Cauchy horizon, these existence results are not applicable in that region.
It is therefore natural to ask whether 
Kerr--Newman spacetimes  admit closed timelike trajectories of charged test particles;
charged-particle motion in this background has been studied extensively (see
e.g. \cite{eva_orbits} and references therein), but, to our knowledge, not from
this point of view.

We answer this question in the affirmative. Our first basic result
(Theorem~\ref{bij}) establishes a one-to-one correspondence between
charge-to-mass ratios $k$ and radii $r_k$ such that, on the equatorial plane
$\{\theta=\pi/2\}$ and inside the positive-$r$ time-machine region, the
future-pointing timelike unit field $-\bigl(-g_{\phi\phi}\bigr)^{-1/2}\,\partial_\phi$
is the four-velocity of a charged test particle that moves, solely under the action of
the Kerr--Newman electromagnetic field, along the closed timelike circle
of radius $r_k$. Equivalently, prescribing the radius of an equatorial
$\partial_\phi$-circle in $\mathfrak T$ singles out the unique charge-to-mass
ratio for which that circle solves the Lorentz force equation. We notice that in the positive $r$ time-machine region, the charge-to-mass ratio must have opposite sign to that of $\bh$.

Beyond these $\partial_\phi$-circles, we study two further families of
equatorial electromagnetic orbits that interact with the time-machine region.
In Theorem~\ref{proposition alpha curve} we prove the existence of
spherical (non-closed) electromagnetic orbits, of constant Boyer--Lindquist radius,
contained in $\mathfrak T$ for both signs of $r$, and we obtain closed-form
expressions for their angular velocity, charge-to-mass ratio, energy and angular
momentum in terms of the radius. 
Although spherical electromagnetic orbits in Kerr--Newman spacetime fall within the general classification of charged-particle motion in \cite{eva_orbits} (see also \cite{semerak1997}), we record here an explicit family in $\mathfrak T$, with the radius determining the charge-to-mass ratio whose sign   must be opposite  to that of $r\bh$ (see Remark~\ref{remark sign spherical charge}). In particular,  such orbits run backwards in Boyer--Lindquist $t$-coordinate, and it is this reversed time orientation that drives the causality-violating construction in Section~\ref{subsection_causality_violation_process}.
In Section~\ref{subsection_flyby orbits} we then
give, for both positive and negative $r$, sufficient  conditions on the
constants of motion ensuring that an equatorial electromagnetic orbit with charge-to-mass ratio having opposite sign to that of $\bh$ is a
flyby orbit (Lemma~\ref{proposition flyby reaching the time-machine}). In comparison with  \cite{eva_orbits}, our purpose in Section~\ref{subsection_flyby orbits} is more specific. Indeed, we also give a result  guaranteeing that an equatorial flyby electromagnetic orbit has its unique  turning point inside 
$\mathfrak T$ (Theorem~\ref{prescribed_common_turning_fyby}).
These results  are tailored to the construction of broken electromagnetic orbits.

By  a {\em broken electromagnetic orbit} we mean a continuous, piecewise smooth worldline whose smooth pieces solve the Lorentz force equation, possibly with different charge-to-mass ratios. Such
curves model decay processes of charged particles into fragments having
different charge-to-mass ratios. In Section~\ref{decay_section} we construct two
successive decays: first of an equatorial flyby particle reaching a radial turning value
$\bar r\in\mathfrak T$ followed by the one  of a particle moving along  an equatorial spherical orbit  with the same radius $\bar r$ and produced at the first decay. We make explicit the conservation laws holding at the decay vertices: conservation
of electric charge together with the mass-weighted conservation of energy and
angular momentum, equivalent to the conservation of the kinetic four-momentum
$\pi$. Building on this, in
Section~\ref{subsection_causality_violation_process} we prove
(Theorem~\ref{closed_broken_orbit_blockIII}) that, both in
$\{\theta=\pi/2\}\cap\{r>0\}$ and in $\{\theta=\pi/2\}\cap\{r<0\}$, there exist
closed broken electromagnetic orbits obtained by concatenating two equatorial flyby
branches with a common radial turning value $\bar r\in\mathfrak T$ and an equatorial  spherical orbit of radius $\bar r$. These closed broken orbits lie inside the
inner horizon and are interpreted as causality-violating decay processes; we
also note in Remark~\ref{brokenKerr} that an  analogous concatenation cannot be realized by timelike
geodesics in sub-extremal Kerr.

A second physical reading of this construction connects it with the Penrose
process \cite{first_Penrose_process,Penrose_process}. In the
Kerr spacetime this energy-extraction mechanism relies on the \emph{ergosphere},
the region between the event horizon and the stationary limit surface where
$\partial_t$ is spacelike: a particle decaying there may produce a fragment in retrograde motion, having negative energy falling into the black hole, while the other fragment escapes to
infinity with energy larger than that of the incoming particle, the excess being
extracted from the rotational energy of the black hole. 
Recent results \cite{Ruffini_PRL} have revisited the efficiency of the
single Penrose process and recovered the energy-extraction picture originally
envisaged by Penrose, while \cite{Ruffini_second} extended the analysis to
repetitive decay processes in the extreme Kerr spacetime (see also
\cite{Ruffini_book} for an overview of the Penrose process in relativistic
astrophysics). The energetic interpretation of this mechanism is closely tied
to Christodoulou's irreducible mass, to the area theorem, and to the laws of
black-hole mechanics
\cite{Christodoulou1970, Christ_Ruffini,Hawking_area,Hawking_thermodynamics,Wald}.
In \cite{Christ_Ruffini}, Christodoulou and Ruffini also observed that for charged particles in presence of an electromagnetic field, negative-energy states do not necessarily require that the decay happens in the ergosphere:
they can arise from the electromagnetic contribution to the conserved Killing
energy (see also \cite{Denardo_Ruffini}). Later, the Kerr--Newman
case was studied in detail in \cite{Prasanna1983, BDD1985}. These works, however, did not relate this phenomenon to the time-machine region. In
Section~\ref{Penrose} we consider an energy extraction  process which takes place in the equatorial section of  $\mathfrak T$, where $\partial_t$ is still
future-pointing timelike, yet the electromagnetic interaction allows one
daughter particle to have negative energy. By conservation of the mass-weighted
energy, the other daughter particle is then produced with energy exceeding that
of the parent particle. Since the decay occurs below the Cauchy horizon, this
higher-energy particle cannot return to the original asymptotic region; instead it can propagate, through the horizons of a subsequent Kerr--Newman patch, to a
different asymptotically flat end of the maximal analytic extension.
In this sense, energy is transferred between distinct asymptotic regions. The
causality-violating process  enters through the complementary branch of the
decay chain: the negative-energy fragment moving along an equatorial  spherical orbit later decays and gives rise
to another equatorial flyby electromagnetic orbit. Explicit numerical examples, listing
the rest masses, electric charges, energies and angular momenta of all the
particles involved, in both the positive- and negative-\(r\) cases, are
collected in Section~\ref{subsection_numerics}.
\subsubsection*{Conventions}
Without loss of generality we will assume the parameter $a$, representing the angular momentum (per unit mass) of the black hole, to be positive. Notice that if $a<0$, one may  redefine the angular coordinate $\phi$ to be $-\phi$, hence obtain the same metric tensor. Moreover, the electric charge $\bh$ of the black hole will be assumed to be non-zero.

\section{Kerr--Newman  spacetime}
Let us consider Boyer--Lindquist coordinates
\[
(t,r)\in \mathbb R^2,
\qquad
(\theta,\phi)\in S^2,
\]
and let us set
\[
\rho^2:=r^2+a^2\cos^2\theta,
\qquad
\Delta:=r^2-2Mr+a^2+\bh^2,
\]
where
\[
a>0,
\qquad
M>0,
\qquad
\bh\in \mathbb R\setminus\{0\}.
\]
We assume throughout the sub-extremal condition
\[
a^2+\bh^2<M^2.
\]
Thus, $\Delta$ has two distinct real roots
\[
r_\pm:=M\pm\sqrt{M^2-a^2-\bh^2}.
\]
The expression of the metric below is singular at the roots of
$\Delta$. We therefore distinguish the Boyer--Lindquist coordinate domain from
the horizon-crossing spacetime. Let
\[
\mathcal M_{\mathrm{BL}}
:=
(\mathbb R^2\times S^2)\setminus(\Sigma
\cup\{r=r_+\}\cup\{r=r_-\}),
\qquad
\Sigma:=\{\rho=0\}=\{r=0,\ \theta=\pi/2\}.
\]
On $\mathcal M_{\mathrm{BL}}$ the Kerr--Newman metric is given by
\begin{align*}
	g={}&-\frac{\Delta-a^2\sin^2\theta}{\rho^2}\,dt^2
	+\frac{2a\sin^2\theta\,[\Delta-(r^2+a^2)]}{\rho^2}\,dt\,d\phi \nonumber\\
	&\quad
	+\frac{\sin^2\theta[(r^2+a^2)^2-a^2\Delta\sin^2\theta]}{\rho^2}\,d\phi^2
	+\frac{\rho^2}{\Delta}\,dr^2+\rho^2\,d\theta^2,
\end{align*}
(see \cite[Appendix D]{FrolovNovikov1998}). We fix once and for all, a smooth horizon-crossing Kerr--Newman extension
$(\mathcal M,g)$ of $(\mathcal M_{\mathrm{BL}},g)$. The equations
\[
\mathscr H_\pm:=\{r=r_\pm\}
\]
define smooth null hypersurfaces only in this extension.  The hypersurface $\mathscr H_+$ is the {\em event horizon} and $\mathscr H_-$ is the {\em Cauchy horizon}.
All Boyer--Lindquist coordinate computations below are understood inside
$\mathcal M_{\mathrm{BL}}$, while statements involving crossing $r=r_\pm$ are
understood in the fixed horizon-crossing extension $\mathcal M$.
\begin{defn}
	The connected components of $\mathcal M_{\mathrm{BL}}$
	\[
	\mathrm{I}:=\{r>r_+\},
	\qquad
	\mathrm{II}:=\{r_-<r<r_+\},
	\qquad
	\mathrm{III}:=\{r<r_-\},
	\]
	are called the Boyer--Lindquist blocks.
\end{defn}
Each Boyer--Lindquist block becomes a time-oriented spacetime when equipped with the following choice of future direction:
\begin{itemize}
	\item on I, the vector field \(-\nabla t\) is declared future-pointing;
	\item on II, the vector field \(-\partial_r\) is declared future-pointing;
	\item on III, the vector field
	\begin{equation*}
	V:=(r^2+a^2)\partial_t+a\partial_\phi
	\end{equation*}
	is declared future-pointing.
\end{itemize}
\begin{oss}\label{horizons}
The choice of future cones in blocks I, II and III is not made independently in each block, but is obtained by continuous extension of the same time orientation across the future horizons. To make this precise, one may pass to ingoing Kerr--Newman coordinates
\begin{equation}\label{regular}
v=t+\int \frac{r^2+a^2}{\Delta}\,dr,
\qquad
\varphi_+=\phi+\int \frac{a}{\Delta}\,dr,
\end{equation}
which are regular across the hypersurfaces $r=r_\pm$. In these coordinates the future event horizon $\mathscr H_+$ and the future Cauchy horizon $\mathscr H_-$ are smooth null hypersurfaces, generated respectively by
\[
K_\pm=\partial_v+\Omega_\pm\,\partial_{\varphi_+},
\qquad
\Omega_\pm=\frac{a}{r_\pm^2+a^2}.
\]
Since the ingoing coordinates are defined by $v=t+f(r)$ and $\varphi_+=\phi+g(r)$, with $r,\theta$ unchanged, the coordinate vector fields $\partial_v$ and $\partial_{\varphi_+}$ coincide with $\partial_t$ and $\partial_\phi$, respectively. Thus $K_\pm=\partial_t+\Omega_\pm\,\partial_\phi$.
Moreover, wherever $\Delta>0$ (that is, in blocks I and III), the vector field $V$ is timelike, since
$g(V,V)=-\Delta\rho^2<0$. At the horizons one has
\[
V\big|_{r=r_\pm}=(r_\pm^2+a^2)\,K_\pm.
\]
In block I the future cone is the connected component of the timelike cone containing $-\nabla t$. Since
$g(V,-\nabla t)=-dt(V)=-(r^2+a^2)<0$,
the vector field $V$ is future-pointing. As $r\to r_+$, it tends to the positive multiple $(r_+^2+a^2)K_+$, hence $K_+$ belongs to the closure of the field of future cones of block I.
In block III the future cones are, by definition, the connected components of the ones  containing $V$. As $r\to r_-$, one has $V\longrightarrow (r_-^2+a^2)K_-$,
so that $K_-$ belongs to the closure  of the field future cones of block III.
In block II one has $\Delta<0$, hence $g_{rr}=\rho^2/\Delta<0$,
so $\partial_r$ is timelike there, and the future cones are  by definition the connected components containing $-\partial_r$. To compare this choice with the horizons, define
\[
X_r^+:=K_+-\sqrt{r_+-r}\,\partial_r
\qquad\text{for }r<r_+\text{ close to }r_+,
\]
and
\[
X_r^-:=K_--\sqrt{r-r_-}\,\partial_r
\qquad\text{for }r>r_-\text{ close to }r_-.
\]
Since $g_{rt}=g_{r\phi}=0$ in block II we have  $g(K_\pm,\partial_r)=0$
and therefore
\[
g(X_r^\pm,X_r^\pm)=g(K_\pm,K_\pm)+|r-r_\pm|\,g_{rr}.
\]
Now $K_\pm$ extend smoothly across the corresponding horizons in ingoing Kerr--Newman coordinates and are null on $r=r_\pm$, so $g(K_\pm,K_\pm)=O(r-r_\pm)$.
On the other hand, for $r\to r_\pm$,
\[
|r-r_\pm|\,g_{rr}\longrightarrow -\frac{\rho_\pm^2}{r_+-r_-}<0.
\]
Hence $X_r^\pm$ are timelike for $r$ sufficiently close to $r_\pm$. Moreover,
\[
g(X_r^\pm,-\partial_r)=\sqrt{|r-r_\pm|}\,g_{rr}<0,
\]
so $X_r^\pm$ are also future-pointing in  block II. Letting $r\to r_\pm$, we conclude that both $K_+$ and $K_-$ belong to the closure of the field future cones of block II.
Therefore the future cones chosen in blocks I, II and III are mutually compatible: they arise from the same continuous time orientation across the future event horizon and the future Cauchy horizon.
\end{oss}	

The Kerr--Newman spacetime was discovered in \cite{NewmanEtAl1965} as the electrovac solution of the Einstein--Maxwell equations describing a rotating charged black hole with mass $M$, angular momentum per unit mass $a$, and electric charge $\bh$.  The  electromagnetic field in Boyer--Lindquist coordinates is $F=dA$
where an electromagnetic potential $A$ is given by
\begin{equation}\label{A}
A=
-\frac{\bh r}{\rho^2}\bigl(dt-a\sin^2\theta\,d\phi\bigr).
\end{equation}
The nonvanishing covariant components of $F$ are: 
\[
F_{tr}=-\frac{\bh\,(r^2-a^2\cos^2\theta)}{\rho^4},
\qquad
F_{t\theta}=\frac{2\bh a^2 r\sin\theta\cos\theta}{\rho^4},
\]
\[
F_{r\phi}=-\frac{\bh a\sin^2\theta\,(r^2-a^2\cos^2\theta)}{\rho^4},
\qquad
F_{\theta\phi}=\frac{2\bh a r(r^2+a^2)\sin\theta\cos\theta}{\rho^4};
\]
they extend smoothly across the horizons. The associated endomorphism $\bar F:T\mathcal M\to T\mathcal M$,  is the fiberwise linear map defined by
\begin{equation}\label{barF}
g(v,\bar F(w))=F(v,w),
\qquad
v,w\in T\mathcal M.
\end{equation}
On $\mathcal M_{\mathrm{BL}}$ its nonvanishing components $F^\alpha{}_\beta=g^{\alpha\mu}F_{\mu\beta}$ are:
\begin{equation}\label{Ftrtheta}
F^t{}_r=\frac{\bh(r^2+a^2)(r^2-a^2\cos^2\theta)}{\rho^4\Delta},
\qquad
F^t{}_\theta=-\frac{2\bh a^2 r\sin\theta\cos\theta}{\rho^4},
\end{equation}
\begin{equation}\label{Frtphi}
F^r{}_t=\frac{\bh\Delta(r^2-a^2\cos^2\theta)}{\rho^6},
\qquad
F^r{}_\phi=-\frac{\bh a\Delta\sin^2\theta\,(r^2-a^2\cos^2\theta)}{\rho^6},
\end{equation}
\begin{equation*}
F^\theta{}_t=-\frac{2\bh a^2 r\sin\theta\cos\theta}{\rho^6},
\qquad
F^\theta{}_\phi=\frac{2\bh a r(r^2+a^2)\sin\theta\cos\theta}{\rho^6},
\end{equation*}
\begin{equation}\label{Fphirtheta}
F^\phi{}_r=-\frac{\bh a(a^2\cos^2\theta-r^2)}{\rho^4\Delta},
\qquad
F^\phi{}_\theta=-\frac{2\bh a r\cot\theta}{\rho^4}.
\end{equation}
\begin{oss}
With these choices for the vector potential $A$ in \eqref{A} and the endomorphism $\bar F$ in \eqref{barF}, our Lorentz force agrees with both Carter \cite{Carter_1966_Axis} and Frolov--Novikov \cite{FrolovNovikov1998}. Indeed, Carter uses the opposite sign for the potential $A$, but writes the Lorentz force with the mixed tensor obtained by raising the second index, which differs by a minus sign from $F^{\alpha}{}_\beta$.
\end{oss}
\subsection{The time-machine region}
Let us introduce the so-called  time-machine region of the Kerr--Newman spacetime in Boyer--Lindquist coordinates.
\begin{defn}
The subset of $\mathcal M$ where the Killing vector field $\partial_\phi$ becomes timelike,
\begin{align*}
\mathfrak{T}=\{p\in\mathcal{M}: g_p(\partial_\phi,\partial_\phi)<0\},
\end{align*}
is called \textit{time-machine region}.
\end{defn}
The time-machine region  $\mathfrak{T}$ contains points with both negative and positive $r$-coordinate if the charge of the black hole $\bh$ is non-vanishing,  while it is contained in $\{r<0\}$ if $\bh=0$,  see \cite[Lemma 2.4.9]{KBH_book}. Indeed, if we restrict to $\{\theta=\pi/2\}$ and $\bh\neq 0$
\begin{align}\label{gphiphi in pi/2}
g(\partial_\phi,\partial_\phi)|_{\theta=\pi/2} = r^2 + \frac{a^2 (r^2+2Mr-\bh^2)}{r^2}\approx -\frac{a^2 \bh^2}{r^2}<0,                                     
\end{align}
for $|r|$ sufficiently small. More precisely, 
there exists a unique $r_{\mathrm{tm}}>0$ such that 
\begin{equation*}
\left\{r(x):\; x\in \mathfrak{T}\cap \{r>0\},\; \theta(x)=\frac{\pi}{2}\right\}=(0,r_{\mathrm{tm}}).
\end{equation*}
In fact, from \eqref{gphiphi in pi/2}
$g_{\phi\phi}(r,\pi/2)=\frac{G(r)}{r^2}$,
where
\begin{equation}\label{G}
G(r):=r^4+a^2r^2+2a^2Mr-a^2\bh^2.
\end{equation}
Since $G(0)=-a^2\bh^2$ and $G'(r)=4r^3+2a^2r+2a^2M>0$ for $r>0$, 
there exists a unique $r_{\mathrm{tm}}>0$ such that $G(r_{\mathrm{tm}})=0$ and
$G(r)<0$ for  $0<r<r_{\mathrm{tm}}$.
\begin{lem}\label{rtm}
	The positive $r$ equatorial time-machine region is contained in the Boyer--Lindquist block $\mathrm{III}$, i.e.
	\[
	0<r_{\mathrm{tm}}<r_-.
	\]
\end{lem}
\begin{proof}
	Since $G(0)=-a^2\bh^2<0$ and $G$ is strictly increasing for $r>0$,  it is enough to show that $G(r_-)>0$. Since $\Delta(r_-)=0$, that is $2Mr_-=r_-^2+a^2+\bh^2$
	we obtain
	\[G(r_-)=
		r_-^4+a^2r_-^2+a^2(r_-^2+a^2+\bh^2)-a^2\bh^2
		=
		(r_-^2+a^2)^2>0.
	\]
\end{proof}
Observe that the Killing vector field $\partial_t$ is always timelike when $r<0$ since
\[
g(\partial_t,\partial_t)=-\frac{r^2-2Mr+a^2\cos^2\theta+\bh^2}{\rho^2}<0,             
\]
(see Figure~\ref{timelikepartial_trneg}), while it changes causal character if $r\in(0,r_{-})$ (see Figures~\ref{timelikepartial_trpos} and \ref{timelikepartial_trpos_zoom}).   When restricted to $\{\theta=\pi/2\}$,  we have

\begin{align}\label{gtt in pi/2}
g(\partial_t,\partial_t)|_{\theta=\pi/2}= -\frac{ r^2-2Mr+\bh^2}{r^2}.  
\end{align}
Therefore, $\partial_t$ is timelike for $r<0$, and on $\{\theta=\pi/2\}$ if $r\in(0,M-\sqrt{M^2-\bh^2})$ and it is spacelike if $r\in(M-\sqrt{M^2-\bh^2},r_{-})$.  Hence,
\begin{align*} 
\{0<r<M-\sqrt{M^2-\bh^2}\}\cap\mathfrak{T}\neq \emptyset,
\end{align*}
so $\partial_t$ can also be timelike in $\mathfrak{T}\cap\{r>0\}$.

\vspace{0.2cm}
\begin{figure}[H]
\centering
\includegraphics[scale=0.4]{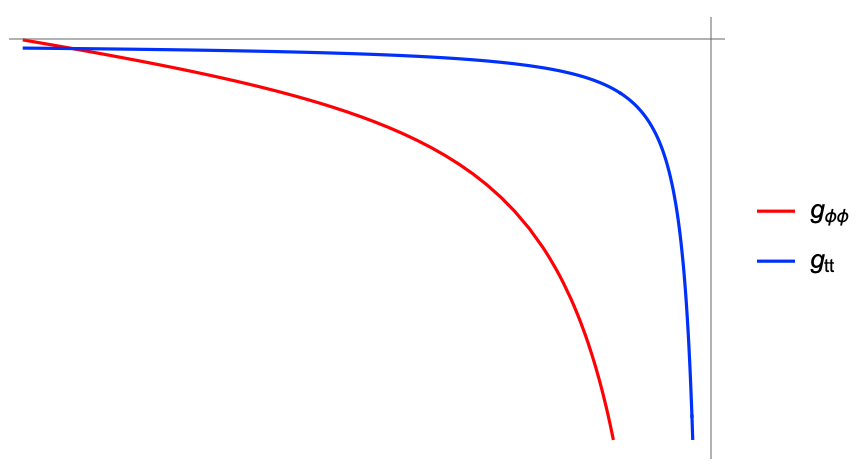} 
\vspace{0.2cm}
\caption{Causal character of $\partial_\phi$ and $\partial_t$ in $\mathfrak{T}\cap\{r<0\}\cap\{\theta=\pi/2\}$. Plot with $a=3,\,  M=5,\,  \bh=1$.}\label{timelikepartial_trneg}
\end{figure}

\vspace{0.2cm}

\begin{figure}[H]
\centering
\includegraphics[scale=0.4]{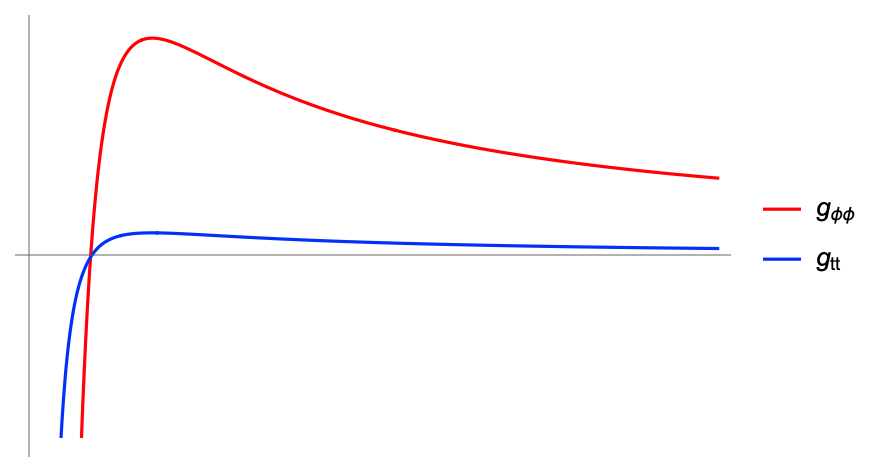} 
\vspace{0.2cm}
\caption{Causal character of $\partial_\phi$ and $\partial_t$ in $\{0<r<r_{-}\}\cap\{\theta=\pi/2\}$.  Plot with $a=3,\,  M=5,\,  \bh=1$.} \label{timelikepartial_trpos}
\end{figure}

\vspace{0.3cm}

\begin{oss} \label{remark_partial_t_timelike_time_machine}
In $\{\theta=\pi/2\}\cap\{0<r<r_{-}\}$,  from Eq.  \eqref{gtt in pi/2} $\partial_t$ is null when $r^2=2Mr-\bh^2$.  Substituting the latter in Eq.  \eqref{gphiphi in pi/2} we get
\[
g(\partial_\phi,\partial_\phi)|_{\theta=\pi/2} = r^2 + 2a^2>0.
\]
So $\partial_\phi$ is spacelike when $\partial_t$ becomes null. Therefore in $\{0<r<r_{-}\}$
\begin{align*}
\mathfrak{T}\subset \{g(\partial_t,\partial_t)|_{\theta=\pi/2}<0\}.
\end{align*}
\begin{figure}[H]
\centering
\includegraphics[scale=0.4]{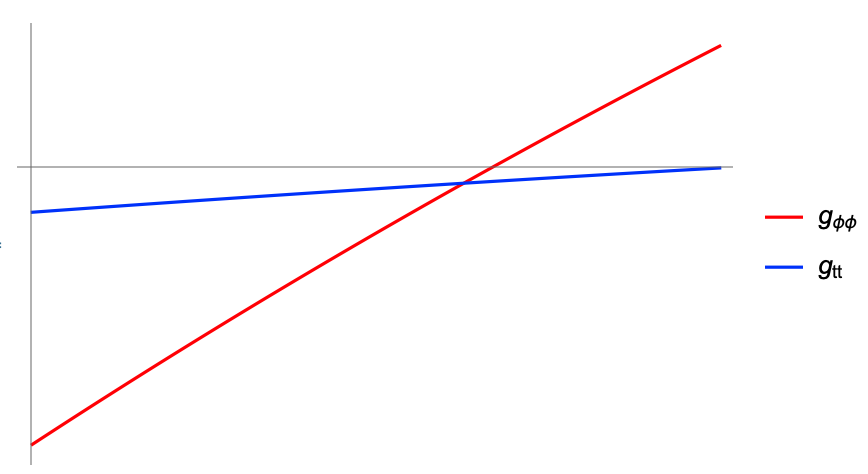} 
\vspace{0.2cm}
\caption{Zoom of Figure~\ref{timelikepartial_trpos} at the intersection point.}\label{timelikepartial_trpos_zoom}
\end{figure}
\end{oss}
\begin{prop}\label{time_orientation_killing_vectors}
In block $\mathrm{III}$ if the Killing vector fields $\partial_t$ and $\partial_\phi$ are causal,  then
\begin{itemize}
\item $\partial_\phi$ is past-pointing;
\item $\partial_t$ is future-pointing.
\end{itemize}
\end{prop}
\begin{proof}
The block III is time-oriented by the canonical vector field $V$ and in this region $\Delta>0$.  Therefore
\begin{align*}
g(\partial_{t},V)&= -\Delta<0,\\
\intertext{and}
g(\partial_{\phi},V)&= a \Delta \sin^2\theta >0,
\end{align*}
if $\theta\neq 0,\pi$.  Notice that if $\theta=0,\pi$,  $\partial_\phi\equiv 0$,  so it is spacelike. 
\end{proof}

\subsection{Equations of motion of charged test particles}
We write
\[
k:=\frac{e}{m}
\]
for the  \textit{charge-to-mass ratio} of a test particle, where $e$ is its electric charge and $m$ its rest mass.
\begin{defn}[Electromagnetic orbits]
A timelike, future-pointing  curve $\gamma$ parametrized by proper time, i.e. 
\[
g(\dot\gamma,\dot\gamma)=-1,
\]
is an {\em electromagnetic orbit} if it  is a trajectory of a charged test particle, i.e.  it satisfies the
\emph{Lorentz force equation} (LFE)
\begin{equation}\label{LFE-abstract}
	\nabla_{\dot\gamma}\dot\gamma = k\,\bar F(\dot\gamma),
\end{equation}
where  $\nabla$ is the Levi--Civita connection of $g$ and $\bar F$ is defined in \eqref{barF}. 
In local coordinates, \eqref{LFE-abstract} becomes
\begin{equation*}
	\ddot x^\mu+\Gamma^\mu_{\sigma\nu}\dot x^\sigma\dot x^\nu
	=
	k\,F^{\mu}{}_\sigma\dot x^\sigma,
\end{equation*}
where \(x^\mu\) are the coordinate components of \(\gamma\).

In particular an electromagnetic orbit $\gamma:[\tau_0,\tau_1]\to  \mathcal M$ will be called \emph{closed} if \begin{align*}
	\gamma(\tau_0)=\gamma(\tau_1),  \quad \dot{\gamma}(\tau_0)=\dot{\gamma}(\tau_1).
\end{align*}
\end{defn}
Since on every Boyer--Lindquist block  $\partial_t$ and $\partial_\phi$ are Killing vector fields and the electromagnetic potential \(A\) is invariant under both symmetries, the quantities
\begin{equation}\label{EL}
E:=-g(\dot\gamma,\partial_t)-kA_t,
\qquad
L:=g(\dot\gamma,\partial_\phi)+kA_\phi,
\end{equation}
are constant along every solution of \eqref{LFE-abstract} contained in a block. For an orbit crossing a horizon, these constants are computed in the adjacent Boyer--Lindquist blocks using the same stationary-axisymmetric
gauge; the Lorentz force equation itself is expressed invariantly in terms of the smooth $2$-form $F$. Thus, they are constant for every orbit in $(\mathcal M, g)$; they are called the
\emph{conserved energy} and \emph{angular momentum} ({\em per unit mass}) respectively.

There is a further constant of motion associated with the Killing tensor field which the Kerr--Newman spacetime admits. More precisely, since the electromagnetic field satisfies the symmetry condition in \cite[Eq. (D.20)]{FrolovNovikov1998}, it follows that 
\[\mathcal{K} = \left( E a \sin \theta - \frac{L}{\sin \theta} \right)^2 + (g_{\theta\theta}\dot\theta)^2 +  a^2 \cos^2 \theta,\]
is constant along a solution $\gamma$ (see \cite[Eq. (D.28)]{FrolovNovikov1998} and take into account that we are considering $\gamma$ parametrized with  proper time, so that our $\mathcal K$ above corresponds to $\mathcal K/m^2$ in \cite{FrolovNovikov1998}). 
The  {\em Carter constant} is then defined by
\begin{equation}\label{carterconst}
\mathcal Q:=\mathcal K-(aE-L)^2.
\end{equation}

Let us also introduce the vector field
\[
W:=a\sin^2\theta\,\partial_t+\partial_\phi,
\]
and define
\begin{equation}\label{PD}
P(r):=-g(\dot\gamma,V)=E(r^2+a^2)-aL-k\,\bh\,r,
\qquad
D(\theta):=g(\dot\gamma,W)=L-aE\sin^2\theta. 
\end{equation}
The following result, originally proved by Carter in \cite{Carter_causality} using Kerr--Newman coordinates, is taken from \cite[\S 3.4.1]{FrolovNovikov1998} (see also \cite[Exe. 33.7]{MTW}).
\begin{thm}\label{first-order-equations}
	Let \(\gamma\) be a future-pointing timelike solution of LFE
	\eqref{LFE-abstract}, parametrized by proper time and contained in a
	Boyer--Lindquist block. Let  $E$, $L$, and $\mathcal Q$ be its constants of motion. Then \(\gamma\) satisfies the first-order system
	\begin{equation}\label{equations-of-motions}
		\begin{cases}
			\rho^2 \dot t
			=
			a\,D(\theta)+(r^2+a^2)\dfrac{P(r)}{\Delta},
			\\[0.4em]
			\rho^2 \dot\phi
			=
			\dfrac{D(\theta)}{\sin^2\theta}
			+\dfrac{a\,P(r)}{\Delta},
			\\[0.6em]
			\rho^4 \dot r^{\,2}=R(r),
			\\[0.3em]
			\rho^4 \dot\theta^{\,2}=\Theta(\theta),
		\end{cases}
	\end{equation}
	where
	\begin{equation}\label{R}
	R(r):=P(r)^2-\Delta(\mathcal K+r^2),
	\end{equation}
	and
	\begin{equation}\label{Theta}
	\Theta(\theta):=
	\mathcal K-a^2\cos^2\theta-\frac{D(\theta)^2}{\sin^2\theta}
	=
	\mathcal Q+\cos^2\theta
	\left[
	a^2(E^2-1)-\frac{L^2}{\sin^2\theta}
	\right].
	\end{equation}
	\end{thm}
\begin{oss}
The  first-order equations in Theorem~\ref{first-order-equations} are naturally obtained for the  four-momentum $\pi^\mu=dx^\mu/d\lambda$, where $g(\pi,\pi)=-m^2$, see e.g. \cite[Appendix D.4.1]{FrolovNovikov1998}. In the timelike case, after passing to proper time $\tau$ and using $\pi^\mu=m\,dx^\mu/d\tau$, one obtains  \eqref{equations-of-motions} in terms of the four-velocity $u^\mu=dx^\mu/d\tau$, with the conserved quantities understood per unit mass.
\end{oss}	
Since the left-hand sides of the third and fourth equations in \eqref{equations-of-motions} are non-negative, the motion can only occur in
the region where
\[
R(r)\ge 0,
\qquad
\Theta(\theta)\ge 0.
\]
This observation allows one to describe the qualitative behaviour of the radial
and angular coordinates. Although the following statement is proved in
\cite[Corollary 4.3.8]{KBH_book} for Kerr geodesics, the same proof applies
to the  LFE in Kerr--Newman, since it only uses the separated
first-order equations for $R$ and $\Theta$ and the smooth dependence of
solutions on the initial data.
\begin{prop}\label{initial-conditions-and-zeroes}
	Suppose \(R(r_0)=0\), and let \(\gamma\) be an electromagnetic orbit such that
	\[
	r(\gamma(\tau_0))=r_0,
	\qquad
	\frac{d}{d\tau}(r\circ\gamma)(\tau_0)=0.
	\]
	Then:
	\begin{enumerate}
		\item if \(R'(r_0)\neq 0\), then \(r_0\) is an \(r\)-turning value, namely
		\((r\circ\gamma)'(\tau)\) changes sign at \(\tau_0\);
		\item if \(R'(r_0)=0\), then \(r(\gamma(\tau))=r_0\) for all \(\tau\), i.e. the electromagnetic orbit has constant radial coordinate.
	\end{enumerate}
\end{prop}
The above proposition leads to the following definition.
\begin{defn}\label{orbits}
		A maximal electromagnetic orbit $\gamma:I\to  \mathcal M$ is called
	\begin{itemize}
		\item a \emph{flyby electromagnetic orbit}
		if there exists $\tau_0\in I$ such that
		\[
		r(\gamma(\tau_0))=\bar r,\qquad
		R(\bar r)=0,\qquad R'(\bar r)\neq0,
		\]
		and 
		\[
		r\circ\gamma(I)=(-\infty,\bar r]
		\quad\text{or}\quad
		r\circ\gamma(I)=[\bar r,+\infty).
		\]
		The number  $\bar r$ is called the radial turning value of $\gamma$.
		\item a \emph{spherical electromagnetic orbit} if $r(\gamma(\tau))=r_0$ for all $\tau$, where $r_0$ is at least a double zero of the corresponding radial polynomial $R(r)$.
        \item a \emph{circular electromagnetic orbit} if it is a spherical electromagnetic orbit and $\theta(\gamma(\tau))=\theta_0$ for all $\tau$, where $\theta_0$ is at least double zero of the corresponding $\theta$-function $\Theta(\theta)$.
        \item an \emph{$r$-bouncing electromagnetic orbit}
		if there exists $\tau_0<\tau_1\in I$ such that
		\[
		r(\gamma(\tau_{0}))=r_{0}<r_1=r(\gamma(\tau_{1})),\qquad
		R(r_{i})=0,\qquad R'(r_{i})\neq0,\quad\text{for $i=0,1$}
		\]
		and 
		\[
		r\circ\gamma(I)=[r_0,r_1].
        \]
		The numbers $r_{i}$, $i\in\{0,1\}$, are called the radial turning values of $\gamma$.
	\end{itemize}
\end{defn}
\begin{oss}\label{extensions}
A flyby electromagnetic orbit may cross the horizons $r=r_\pm$.  Such a crossing	is not described by the Boyer--Lindquist chart itself, but by the	horizon-crossing extension $\mathcal M$, for instance using the ingoing	coordinates \eqref{regular} near future horizons or the corresponding outgoing	coordinates near past horizons. The radial coordinate $r$ is unchanged in these	charts. Thus the preceding definition is naturally understood in	$\mathcal M$, while the separated first-order equations	\eqref{equations-of-motions} are used inside individual 
Boyer--Lindquist blocks.  
\end{oss}
\section{Electromagnetic orbits} 
In this section we present several different results about electromagnetic orbits. 
We  restrict the dynamics to the equatorial hyperplane $\mathcal E:=\{\theta=\pi/2\}$;  this is dynamically consistent, since $\mathcal E$ is invariant under the Lorentz-force flow:
let $\mathscr R\colon\mathcal M\to\mathcal M$ be the equatorial reflection
$\mathscr R(t,r,\theta,\phi):=(t,r,\pi-\theta,\phi)$, whose fixed-point set is $\mathcal E$. 	Since the metric coefficients depend on $\theta$ only through $\sin^2\theta$ and $\cos^2\theta$,
the map $\mathscr R$ is an isometry; moreover $\mathscr R^*A=A$ (the potential \eqref{A} has no
$d\theta$-component, and $\sin^2\theta$, $\rho^2$ are $\mathscr R$-invariant), hence
$\mathscr R^*F=F$. Therefore, if $\gamma$ solves \eqref{LFE-abstract} with charge-to-mass ratio
$k$, so does $\mathscr R\circ\gamma$, with the same $k$, and  if  $\gamma$ is  a solution with $\gamma(0)\in\mathcal E$ and
$\dot\gamma(0)\in T_{\gamma(0)}\mathcal E$,  $\mathscr R\circ\gamma$ shares the same
initial position and velocity. By uniqueness for \eqref{LFE-abstract} they coincide, so $\gamma$
is fixed by $\mathscr R$ and thus $\theta\circ\gamma\equiv\pi/2$.
This confinement is precisely the property defining a massive particle  in
\cite{CaponioGerminarioMasiello2026} (see also \cite{KobialkoBogushGaltsov2022}). We also note, that if $\gamma$ is an equatorial electromagnetic orbit then  $\mathcal Q=0$ (see \eqref{carterconst}).

In the next subsection  we show  that orbits of the unit vector field associated with   $\partial_\phi$ in the time-machine region are electromagnetic for a charge-to-mass ratio uniquely determined by the radius of the orbit.
\subsection{Smooth closed electromagnetic orbits  tangent to \texorpdfstring{$\partial_\phi$}{d/dphi}}
Taking into account Proposition~\ref{time_orientation_killing_vectors} and Lemma~\ref{rtm}, in the positive $r$ time-machine region the future-pointing unit tangent along an equatorial $\partial_\phi$-circle is
\[
U=-(-g_{\phi\phi})^{-1/2}\partial_\phi.
\]
Let us define the effective potential
\[
V_k:=\sqrt{-g_{\phi\phi}}+kA_\phi.
\]
\begin{lem}
	Let $p\in\mathfrak T$. 
	Then $U_p$ is a future-pointing timelike vector, and the integral curve of $U$ through $p$
	solves \eqref{LFE-abstract} if and only if	$dV_k(p)=0$.
	\end{lem}
\begin{proof}
	Set  $f:=(-g_{\phi\phi})^{-1/2}$, so that $U=-f\partial_\phi$. Since $\partial_\phi$ is Killing and $\partial_\phi(f)=0$,
	\[
	\nabla_U U=\nabla_{-f\partial_\phi}(-f\partial_\phi)=f^2\nabla_{\partial_\phi}\partial_\phi
	=-\frac12 f^2 \nabla(g_{\phi\phi})
	=
	f\,\nabla\!\bigl(\sqrt{-g_{\phi\phi}}\bigr).
	\]
	On the other hand,
	\[
	\bar F(U)=
	-f\bar F(\partial_\phi)
	=
	-f\,g^{\alpha\mu}F_{\mu\phi}\partial_\alpha
	\]
	Because $A$ is independent of $\phi$,
	$F_{\mu\phi}=\partial_\mu A_\phi$,
	hence
	$
	\bar F(U) = -f\nabla A_\phi$.
	Therefore, the LFE is equivalent to
	\[
	f\,\nabla\!\bigl(\sqrt{-g_{\phi\phi}}\bigr)
	=
	-kf\,\nabla(A_\phi),
	\]
	that is $\nabla\!\bigl(\sqrt{-g_{\phi\phi}}+kA_\phi\bigr)=0$.
\end{proof}
On the equatorial plane, recalling \eqref{A} and \eqref{gphiphi in pi/2}, we have
\[
V_k=V_k(r)=
\frac{\sqrt{-G(r)}+ka\bh}{r}
\]
and 
\[
V_k'(r)
=
\frac{-ka\bh\sqrt{-G(r)}-\bigl(a^2\bh^2-a^2Mr+r^4\bigr)}{r^2\sqrt{-G(r)}}.
\]
Hence,
\[
V_k'(r)=0
\iff
k=K(r),
\]
where
\[
K(r)
:=
-\frac{a^2\bh^2-a^2Mr+r^4}{a\bh\sqrt{-G(r)}}.
\]
\begin{thm}\label{bij}
	Let $\bar t\in \R$, $\bar\phi\in[0,2\pi)$. Then 
	\begin{enumerate}
		\item for every $k\in\R$ with $\sgn(k)=-\sgn(\bh)$ and $|k|>1$, there exists a unique $r_k\in(0,r_{\mathrm{tm}})$ such that the integral curve of $U$ through $(\bar t, r_k,\pi/2,\bar\phi)$ is a  closed electromagnetic orbit  with charge-to-mass ratio $k$ (so that $K(r_k)=k$);
		\item for every $\bar r\in(0,r_{\mathrm{tm}})$, the equatorial closed timelike circle
		\[
		\alpha_{\bar r}(\tau)
		=
		\bigl(\bar t,\bar r,\tfrac{\pi}{2},\bar\phi+\omega(\bar r)\tau\bigr),
		\qquad
		\omega(\bar r)
		=
		-\bigl(-g_{\phi\phi}(\bar r,\pi/2)\bigr)^{-1/2},
		\]
		is a closed electromagnetic orbit with charge-to-mass ratio
		\begin{equation}\label{k}
		k=K(\bar r)
		=
		\frac{\omega(\bar r)\,\bar r^2}{2a\bh}\,
		\partial_r g_{\phi\phi}(\bar r,\pi/2).
		\end{equation}
	\end{enumerate}
\end{thm}
\begin{proof}
	Let us define
	\begin{equation}\label{fr}
	f(r):=\frac{a^2\bh^2-a^2Mr+r^4}{\sqrt{-G(r)}}.
	\end{equation}
	Then $f(0)=a|\bh|$; moreover, as $r_{\mathrm{tm}}$ is the unique zero of $G(r)$ for $r\in [0, +\infty)$ (recall Lemma~\ref{rtm}), we have 
		\[
		a^2\bh^2=2a^2Mr_{\mathrm{tm}}+a^2r_{\mathrm{tm}}^2+r_{\mathrm{tm}}^4.
		\]
		Hence
		\[
		a^2\bh^2-a^2Mr_{\mathrm{tm}}+r_{\mathrm{tm}}^4
		=
		a^2Mr_{\mathrm{tm}}+a^2r_{\mathrm{tm}}^2+2r_{\mathrm{tm}}^4>0.
		\]
		Therefore the numerator of $f$ converges to a strictly positive limit, and thus
		\[
		\lim_{r\to r_{\mathrm{tm}}^-}f(r)=+\infty.
		\]
		In order to evaluate the derivative of $f$ let us set
		\[
		N(r):=a^2\bh^2-a^2Mr+r^4,
		\qquad
		S(r):=-G(r)=a^2\bh^2-2a^2Mr-a^2r^2-r^4.
		\]
		Then \(S(r)>0\) on \((0,r_{\mathrm{tm}})\), and $f(r)=N(r)/\sqrt{S(r)}$.
		Therefore,
		\[
		f'(r)
		=
		\frac{N'(r)S(r)-\frac12 N(r)S'(r)}{S(r)^{3/2}}
		=
		\frac{(-a^2M+4r^3)S(r)+(a^2M+a^2r+2r^3)N(r)}{S(r)^{3/2}}.
		\]
		Expanding the numerator, one finds
		\[
		(-a^2M+4r^3)S(r)+(a^2M+a^2r+2r^3)N(r)
		=
		r\Bigl(a^4(M^2+\bh^2)+6a^2\bh^2r^2-8Ma^2r^3-3a^2r^4-2r^6\Bigr).
		\]
		Now, using the definition of \(S(r)\), we rewrite
		\[
		6a^2\bh^2r^2-8Ma^2r^3-3a^2r^4-2r^6
		=
		6r^2S(r)+4Ma^2r^3+3a^2r^4+4r^6.
		\]
		Hence
		\[
		f'(r)
		=
		\frac{r}{S(r)^{3/2}}
		\Bigl(
		a^4(M^2+\bh^2)+6r^2S(r)+4Ma^2r^3+3a^2r^4+4r^6
		\Bigr).
		\]
		Since \(r>0\) and \(S(r)>0\) on \((0,r_{\mathrm{tm}})\), it follows that
		$f'(r)>0$ for every $r\in(0,r_{\mathrm{tm}})$.
	Hence, $f$ is strictly increasing and
	\[
	f\bigl((0,r_{\mathrm{tm}})\bigr)=(a|\bh|,+\infty).
	\]
	Since $K(r)=-f(r)/(a\bh)$, the function $K$ is one-to-one. In particular \[K\big((0,r_{\mathrm{tm}})\big)=\begin{cases}
	 (-\infty, -1)&\text{if    $\bh>0$}\\
    (1,+\infty)& \text{if    $\bh<0$}
	\end{cases}\]
    and this proves statement (1).
	
	To prove (2), 	we have to show that the charge-to-mass ratio $k$ in \eqref{k}
	coincides with $K(\bar r)$. By \eqref{gphiphi in pi/2},
	$g_{\phi\phi}(r,\pi/2)=\frac{G(r)}{r^2}$ and we have 
	on $(0,r_{\mathrm{tm}})$,
	\[
	\omega(r)=-\frac{r}{\sqrt{-G(r)}}.
	\]
	By direct computation, we can see that
	\[
	\partial_r g_{\phi\phi}(r,\pi/2)
	=
	2r-\frac{2a^2M}{r^2}+\frac{2a^2\bh^2}{r^3},
	\]
	thus
	\[
	\frac{\omega(r)\,r^2}{2a\bh}\,\partial_r g_{\phi\phi}(r,\pi/2)
	=
	-\frac{r^4-a^2Mr+a^2\bh^2}{a\bh\sqrt{-G(r)}}
	=
	K(r).
	\]
	\end{proof}
\begin{oss}
The proper-time period of $\alpha_{\bar r}$ is
\[T_{\bar r}=2\pi\sqrt{-g_{\phi\phi}(\bar r,\pi/2)}
=
\frac{2\pi}{\bar r}\sqrt{a^2Q^2-2a^2M\bar r-a^2\bar r^2-\bar r^4}.
\]	
\end{oss}
\begin{oss}
The two parts of Theorem~\ref{bij} are equivalent formulations of the same result. The first fixes the charge-to-mass ratio and determines the unique radius of the corresponding equatorial closed timelike circle. The second fixes the radius of an equatorial closed timelike circle and determines the unique charge-to-mass ratio for which it is a solution of the Lorentz force equation.
	\end{oss}
    \begin{oss}
 The interaction between the black hole charge $\bh$ and the particle’s charge (per unit mass) \eqref{k}  provides precisely the force required for  the particle  to follow a smooth closed solution of the Lorentz force equation. Notice that, in contrast, the existence of closed timelike curves in the Kerr spacetime would require an additional external force acting on the particles, since no electromagnetic field source is present.
\end{oss}

\begin{cor}
	Let $k$ and $r_k$ be as in Theorem~\ref{bij}. Then the corresponding smooth closed  electromagnetic orbit satisfies 
		\[
	\mathcal{Q}=0,
	\qquad
	R(r_k)=0,
	\qquad
	R'(r_k)=0.
	\]
	Moreover, the conserved energy and angular momentum per unit mass are
	\begin{equation}\label{ELU}
	E=-\frac{a^2M+r_k^3}{a\sqrt{-G(r_k)}},
	\qquad
	L=-\frac{2r_k^3+a^2(M+r_k)}{\sqrt{-G(r_k)}}.
	\end{equation}
\end{cor}
\begin{proof}
As $\dot\theta=0$, last equation in  \eqref{equations-of-motions} and \eqref{Theta} give
 $\mathcal{Q}=0$.
Using the definition of $E$ and $L$ in \eqref{EL}  we obtain \eqref{ELU}.
Moreover, since the electromagnetic orbit has constant radial coordinate $r(\tau)\equiv r_k$, one has $\dot r=0$,
and the third equation in \eqref{equations-of-motions} yields $R(r_k)=0$.
If $R'(r_k)\neq 0$, then by Proposition~\ref{initial-conditions-and-zeroes},
$r_k$ would be an $r$-turning value, and $(r\circ\gamma)'$ would change sign at the
corresponding parameter value. This is impossible, because $r\circ\gamma$ is constant.
Therefore
$R'(r_k)=0$.
\end{proof}
\begin{prop}\label{radialmax}
	Let $k$ be in the admissible range of Theorem~\ref{bij}, and let $r_k$ be the corresponding radius. Then
	\[
	V_k''(r_k)<0.
	\]
\end{prop}
\begin{proof}
Since
\[
V_k'(r)=\frac{-ka\bh\sqrt{-G(r)}-\bigl(a^2\bh^2-a^2Mr+r^4\bigr)}{r^2\sqrt{-G(r)}},
\]
one may rewrite it as
\[
V_k'(r)=\frac{-a\bh}{r^2}\bigl(k-K(r)\bigr).
\]
Differentiating gives
\[
V_k''(r)
=
\frac{2a\bh}{r^3}\bigl(k-K(r)\bigr)+\frac{a\bh}{r^2}K'(r).
\]
At the critical radius $r=r_k$, where $K(r_k)=k$, this becomes
\[
V_k''(r_k)
=
\frac{a\bh}{r_k^2}K'(r_k)
= -\frac{f'(r_k)}{r_k^2},
\]
$f$ defined in \eqref{fr}. As shown in the proof of Theorem~\ref{bij}, $f'(r_k)>0$, thus $V_k''(r_k)<0$.
\end{proof}
\begin{oss}
	Proposition~\ref{radialmax}  shows that  the closed electromagnetic orbits in Theorem~\ref{bij}  are associated with strict local maximum of the effective potential $V_k$. This does not imply that they are radially stable, i.e. that they also satisfy $R''(r_k)<0$.
\end{oss}

\subsection{Equatorial circular electromagnetic orbits inside the time-machine region} 
In this subsection, we give a necessary and sufficient condition ensuring that a non-closed equatorial circular curve inside the time-machine region is an electromagnetic orbit. 
\begin{thm}\label{proposition alpha curve}
	Let the surface $\mathcal S:=\{r=\bar r,\theta=\pi/2\}$ be contained in $\mathfrak T$, and let
	\[
	\gamma(s)=(t_0-s,\bar r,\pi/2,\phi_0+\omega s),
	\qquad t_0,\omega\in\mathbb R,\ \phi_0\in[0,2\pi).
	\]
	Then $\gamma$ is an electromagnetic orbit if and only if
	\begin{align}
		\omega=\frac{1}{-a+\sqrt{\frac{\bar r^2 (a^2 + \bar r^2)}{\bh^2 - 2 M \bar r}}}. \label{omega middle curve}
	\end{align}
	and
	\begin{align}
		k= -\frac{\big[2\bar r^2 (\bh^2 - \frac{3}{2} M \bar r) + a^2 (\bh^2 - M \bar r)\big]\big[ \bar r^4+ a^2 \bar r^2 + a \sqrt{\bar r^2 (\bh^2 - 2 M \bar r) (a^2 + \bar r^2)}\big]}{\bh \bar r (a^2 + \bar r^2) \big[-\bar r^4 + a^2 \big(\bh^2 - \bar r (2 M + \bar r)\big)\big]}. \label{kbarr} 
	\end{align}
\end{thm}
\begin{proof}
	From \cite[Appendix D.2]{FrolovNovikov1998}  and \eqref{Ftrtheta}--\eqref{Fphirtheta}, we have 
	\[
	\Gamma^\mu_{tt}=\Gamma^\mu_{t\phi}=\Gamma^\mu_{\phi\phi}=0,
	\qquad
	F^\mu{}_\nu\,\dot\gamma^\nu=0,
	\qquad
	\mu=t,\phi.
	\]
	and for $\theta=\pi/2$ also
	\[
	\Gamma^\theta_{tt}=\Gamma^\theta_{t\phi}=\Gamma^\theta_{\phi\phi}=0,
	\qquad
	F^\theta{}_\nu\,\dot\gamma^\nu=0
	\]
	along $\gamma$. Therefore only the $r$-component of the Lorentz force equation is nontrivial.
	Using again  \cite[Appendix D.2]{FrolovNovikov1998} and \eqref{Frtphi}, the $r$-component reads
	\begin{align*}
		\Gamma^r_{tt} (\dot{\gamma}^t)^2
		+\Gamma^r_{\phi\phi} (\dot{\gamma}^{\phi})^2
		+2\Gamma^r_{t\phi} \dot{\gamma}^t \dot{\gamma}^\phi
		=
		k\bigl(F^r{}_t\dot{\gamma}^t + F^{r}{}_\phi \dot{\gamma}^{\phi}\bigr),
	\end{align*}
	that is, 
	\begin{align}\label{LFE of middle curve}
		\frac{\Delta}{\bar r ^5}
		\Bigl[-\bh k \bar r (1+a\omega)	+(\bh+a\bh\omega)^2 +\bar r\bigl(\bar r^3\omega^2-M(1+a\omega)^2\bigr)\Bigr]=0.
	\end{align}
	Solving \eqref{LFE of middle curve} together with the normalization condition $g(\dot\gamma,\dot\gamma)=-1$ gives the two branches 
	\begin{align*}
		\omega_\pm&=\frac{1}{-a\pm \sqrt{\frac{\bar r^2(a^2+\bar r^2)}{\bh^2-2M\bar r}}},\\
		k_\pm&=-\frac{\bigl[2\bar r^2(\bh^2-\frac32 M\bar r)+a^2(\bh^2-M\bar r)\bigr]
			\bigl[\bar r^4+a^2\bar r^2\pm a\sqrt{\bar r^2(\bh^2-2M\bar r)(a^2+\bar r^2)}\bigr]}
		{\bh \bar r (a^2+\bar r^2)\bigl[-\bar r^4+a^2(\bh^2-\bar r(2M+\bar r))\bigr]}.
	\end{align*}
	These quantities are well-defined  because, at points of $\mathcal S$, by Proposition~\ref{time_orientation_killing_vectors},
	\[
	g(\partial_t,\partial_\phi)	= \frac{a(\bh^2-2M\bar r)}{\bar r^2}>0,
	\]
	hence
	\[
	\bh^2-2M\bar r>0,
	\]
	and moreover  $g(\partial_\phi,\partial_\phi)|_{\mathcal S}<0$, i.e.
	\begin{align}\label{partial_phi_timelike}
		\bar{r}^2(\bar{r}^2+a^2)< a^2 (\bh^2 - 2 M \bar{r}).
	\end{align}
	Therefore,
	\[
	0<\sqrt{\frac{\bar r^2(a^2+\bar r^2)}{\bh^2-2M\bar r}}<a,
	\]
	and 
	\begin{equation}\label{1+aomega}
	1+a\omega_+<0,
	\qquad
	1+a\omega_->0.
	\end{equation}
	Hence, $g(\dot\gamma,V)=\Delta(1+a\omega_\pm)$
	is negative only for the $+$-branch. Thus the future-pointing condition selects uniquely $\omega=\omega_+$ and $k=k_+$. 
\end{proof}
\begin{oss}\label{elixorbit}
Theorem~\ref{proposition alpha curve}  can also be interpreted  as an existence result. Indeed, for each point 
	$q\in\mathcal S\subset \mathfrak T$  if the charge-to-mass ratio is chosen as
	$k(\bar r)$ in \eqref{kbarr} and the initial velocity is chosen as
	$-\partial_t+\omega(\bar r)\partial_\phi$, with $\omega(\bar r)$ given by
	\eqref{omega middle curve}, then the unique Lorentz-force trajectory with these
	initial data is the equatorial circular electromagnetic orbit displayed above. We emphasize that $\bar r$ can be positive or negative (see Figures~\ref{elicarpos} and \ref{elicarneg}).
\end{oss}
\begin{cor} \label{cor energy_ang spherical orbit}
	Let $\mathcal S$ and $\gamma$ be as in Theorem~\ref{proposition alpha curve} with $\omega$ and $k$ given by \eqref{omega middle curve} and \eqref{kbarr}.  Then $\gamma$ satisfies 
	\[
	\mathcal{Q}=0,
	\qquad
	R(\bar r)=0,
	\qquad
	R'(\bar r)=0.
	\]
	Moreover, the conserved energy and angular momentum per unit mass are
	\begin{align}
		E&=	\frac{a^2 \bar r (M-\bar r)+a\sqrt{\bar r^2(\bh^2-2M\bar r)(a^2+\bar r^2)}-\bar r^2\bigl[\bar r(M+\bar r)-\bh^2\bigr]}{\bar r^4+a^2\bar r^2-a\sqrt{\bar r^2(\bh^2-2M\bar r)(a^2+\bar r^2)}}, \label{energy middle curve}\\[4pt]
		L&=
		\frac{a^3M\bar r+a\bar r^2(\bh^2-M\bar r)+(a^2+\bar r^2)\sqrt{\bar r^2(\bh^2-2M\bar r)(a^2+\bar r^2)}}
		{\bar r^4+a^2\bar r^2-a\sqrt{\bar r^2(\bh^2-2M\bar r)(a^2+\bar r^2)}}. \label{ang momentum middle curve}
	\end{align}
\end{cor}
\begin{proof}
	Since $\theta\equiv\pi/2$, the last equation in \eqref{equations-of-motions} together with \eqref{Theta} gives
	$\mathcal Q=0$. Moreover, since $r\equiv\bar r$, one has $\dot r=0$, and the third equation in \eqref{equations-of-motions} yields
	$R(\bar r)=0$.
	If $R'(\bar r)\neq 0$, then by Proposition~\ref{initial-conditions-and-zeroes} the value $\bar r$ would be an $r$-turning value, and $(r\circ\gamma)'$ would change sign. This is impossible, because $r\circ\gamma$ is constant. Therefore, 
	$R'(\bar r)=0$.
	Finally, using the definitions of $E$ and $L$ in \eqref{EL}, together with $\dot\gamma^t=-1$, $\dot\gamma^\phi=\omega$ and the values of $\omega$ and $k$ given by \eqref{omega middle curve} and \eqref{kbarr}, we obtain \eqref{energy middle curve} and \eqref{ang momentum middle curve}.
\end{proof}
\begin{oss} \label{remark sign spherical charge}
From \eqref{kbarr} we deduce 
\begin{align*}
\sgn{k}=-\sgn{(\bar{r}\bh)},
\end{align*}
since 
\[
-\bar{r}^4 + a^2 \big(\bh^2 - \bar{r} (2 M + \bar{r})\big)>0
\]
by Eq.  \eqref{partial_phi_timelike} and, if $\bar{r}>0$,
\[
\bh^2 - M \bar{r} > \bh^2 - \frac{3}{2}M \bar{r} > \bh^2 - 2 M \bar{r} >0,
\]
while if $\bar{r}<0$, it is obvious.
Notice that in the region $r>0$, where the gravitational field acts attractively, the Lorentz force must contribute with the same attractive sign; while in the region $r<0$, where the Kerr field is naturally viewed as effectively repulsive, the Lorentz force must again have the same effective sign, namely a repulsive one.
\end{oss}
\begin{figure}[H]
\centering
\includegraphics[scale=0.4]{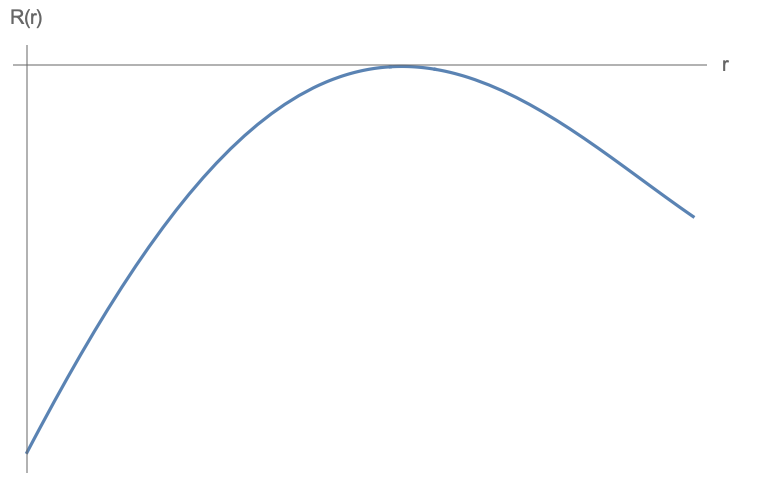}  
\caption{A root of the polynomial $R$ with  $a=3,\,  M=5,\,  \bh=3.5$, corresponding to an equatorial circular electromagnetic orbit with $\bar r=0.9$, in $\mathfrak{T}\cap\{\theta=\pi/2\}$,   hence $k\approx - 2.771$ from Eq.  \eqref{kbarr},  $E\approx -7.408$ from Eq.  \eqref{energy middle curve},  $L\approx -26.055$ from Eq.  \eqref{ang momentum middle curve},  $\mathcal Q=0$.}\label{elicarpos}
\vspace*{-5mm}
\end{figure}
\vspace{0.3cm}
\begin{figure}[H]
\centering
\includegraphics[scale=0.4]{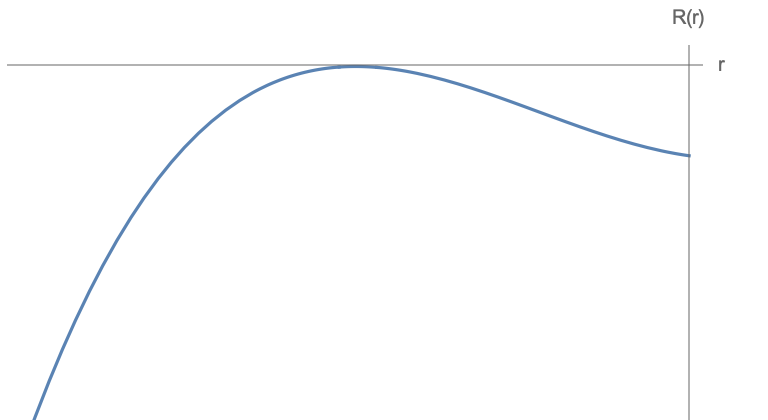} 
\caption{A root of the polynomial $R$ with $a=3,\,  M=5,\,  \bh=1$,  corresponding to an equatorial circular electromagnetic orbit with $\bar r=-1$, in  $\mathfrak{T}\cap\{\theta=\pi/2\}$, hence $k\approx 3.308$ from Eq.  \eqref{kbarr},  $E\approx 0.817$ from Eq.  \eqref{energy middle curve},  $L\approx 0.565$ from Eq.  \eqref{ang momentum middle curve},  $\mathcal Q=0$.} \label{elicarneg}
\vspace*{-5mm}
\end{figure}

\bigskip

\subsection{Flyby electromagnetic orbits with turning point in the time-machine region}\label{subsection_flyby orbits}
We next study equatorial electromagnetic orbits contained either in the
positive or in the negative $r$-region. The two cases are governed by the same
radial polynomial, but the odd powers of $r$ change sign when one passes from
$r>0$ to $r<0$. In order to obtain a   unified formulation, let 
$\sigma\in\{+1,-1\}$, and $\Omega_\sigma:=\{x\in \mathcal M:\sigma r(x)>0\}$, so that  $\Omega_+=\{x\in \mathcal M:r(x)>0\}$ and $\Omega_-=\{x\in \mathcal M:r(x)<0\}$. Let us set 
\begin{equation}\label{AABB}
A:=a^2(E^2-1)-L^2-\bh^2(1-k^2),\qquad	B:=2M(L-aE)^2+2a(L-aE)k\bh.
\end{equation}
\begin{lem}\label{proposition flyby reaching the time-machine}
Let $\gamma$ be a maximal equatorial electromagnetic orbit  in the region
$\Omega_\sigma$, with	constants of motion $(E,L)$ and charge-to-mass ratio $k$. Assume that
$\sgn k=-\sgn\bh$ and $\sigma B<0$.
In addition, assume
\begin{align*}
	&E\geq 1,\qquad\qquad\qquad\text{if } \sigma=+1,\\
\intertext{and} 
	&E>1, \quad A<0,\qquad \text{if $\sigma=-1$.}
\end{align*}
Then  $R$ in \eqref{R} has exactly one simple zero in $\Omega_\sigma$,
denoted $r_{\mathrm{turn}}$, and therefore $\gamma$ is a flyby electromagnetic
orbit in $\Omega_\sigma$ with radial turning value $r_{\mathrm{turn}}$.
In the negative case $\sigma=-1$, one also has the estimate
\[
r_{\mathrm{turn}}<-\frac{2M}{E^2-1}<0.
\]
\end{lem}

\begin{proof}
Since $\gamma$ is equatorial, $\mathcal Q=0$, and $R$ in  \eqref{R} is
\begin{equation*}
R(r)=	(E^2-1)r^4	+2(M-k\bh E)r^3	+Ar^2	+Br	-\bh^2(L-aE)^2.
\end{equation*}
Set $s:=\sigma r>0$ and define $R_\sigma(s):=R(\sigma s)$.
Since $\sgn k=-\sgn\bh$, we have $k\bh<0$. Hence, for $E>0$, we have
$M-k\bh E>0$.	Moreover, 
\[
R_\sigma(0)=-\bh^2(L-aE)^2<0,
\]
because $\sigma B<0$ implies $L\neq aE$.
We now distinguish the two signs of $\sigma$.
	
If $\sigma=+1$, by assumptions we have $E\geq1$ and $B<0$. If $E>1$, the signs of the nonzero
coefficients of $R_+$ are
\[
+\quad +\quad \sgn A\quad -\quad -.
\]
If $E=1$, the quartic term is absent and the signs are
\[
+\quad \sgn A\quad -\quad -.
\]
In either case there is exactly one sign change, independently of the sign of $A$. 

If $\sigma=-1$, by assumptions  $E>1$, $A<0$, and $B>0$. Therefore the signs of the
nonzero coefficients of $R_-$ are
\[
+\quad -\quad -\quad -\quad -.
\]
Again, there is exactly one sign change.
	
By Descartes' rule of signs, $R_\sigma$ has at most one positive zero. Moreover,
\[
R_\sigma(s)\to+\infty	\qquad\text{as }s\to+\infty
\]
(for $\sigma=+1$ this is clear if $E>1$, while if $E=1$ it follows from
$2(M-k\bh)>0$. For $\sigma=-1$ it follows from $E>1$).
Thus $R_\sigma$ has exactly one positive zero. Since Descartes' rule counts
multiplicity, this zero is simple. Equivalently, $R$ has exactly one simple
zero in $\Omega_\sigma$, denoted $r_{\mathrm{turn}}$. By Proposition~\ref{initial-conditions-and-zeroes}, $r_{\mathrm{turn}}$ is an $r$-turning value. As $\gamma$ is maximal, it follows that $r\circ\gamma(I)=[r_{\mathrm{turn}},+\infty)$ if $\sigma=+1$ and $r\circ\gamma(I)=(-\infty,r_{\mathrm{turn}}]$ if $\sigma=-1$. Hence, $\gamma$ is a flyby electromagnetic orbit in $\Omega_\sigma$ with radial turning value $r_{\mathrm{turn}}$.
It remains only to prove the additional estimate in the negative case.  We write
\begin{equation}\label{Rr}
R(r)	=	\big[(E^2-1)r^4+2Mr^3\big]+F(r),
\end{equation}
where
\[
F(r)	=	-2k\bh E r^3	+Ar^2	+Br	-\bh^2(L-aE)^2.
\]
For $r<0$, every term in $F(r)$ is strictly negative. Moreover,
\[	
(E^2-1)r^4+2Mr^3	=	r^3\big((E^2-1)r+2M\big)
\]
Therefore $R(r)<0$ on $\left[-\frac{2M}{E^2-1},0\right)$,
so the negative zero cannot lie in this interval. 
\end{proof}
\begin{oss}\label{insideOmegasigma}
In Lemma~\ref{proposition flyby reaching the time-machine}, the maximal equatorial  electromagnetic orbits $\gamma$ is assumed to be contained in $\Omega_\sigma$.  We emphasize that 
	$\sigma B<0$ implies $L\neq aE$. Hence $\gamma$ does not
	satisfy the necessary condition for reaching the ring singularity in the Kerr--Newman spacetime. Indeed, an electromagnetic orbit crashing into the	ring singularity must be equatorial and satisfy $L=aE$,	see \cite[\S 3.D]{Carter_causality}.  
\end{oss}
\begin{oss}\label{remark flyby existence}
Fix constants $E,L,k$, with $\mathcal Q=0$, such that the hypotheses in Lemma~\ref{proposition flyby reaching the time-machine} hold. 
At a point $p=(t_0,r_{\mathrm{turn}},\pi/2,\phi_0)$
one chooses a future-pointing initial vector $u\in T_p\mathcal M$ with
$u^r=u^\theta=0$, and with $u^t,u^\phi$ determined from  \eqref{EL}  by the prescribed
values of $E$ and $L$. 
The standard  existence and uniqueness theorem for the Lorentz force equation then
gives a unique maximal equatorial electromagnetic orbit with the above initial data which is then a 	flyby  orbit with radial turning value $r_{\mathrm{turn}}$.
\end{oss}

A natural question   is whether one  can construct  flyby equatorial electromagnetic orbits with a prescribed radial turning value  $r_{\mathrm{turn}}$ in the equatorial $r$-section of the time-machine region,   while allowing different charge-to-mass ratios $k$ and different constants of motion $(E,L)$. Thus, these orbits are not merely 
images of a single flyby under the Killing flows generated
by $\partial_t$ or $\partial_\phi$, nor  can they be   reparametrizations either, since throughout we consider only solutions $\gamma$ normalized by $g(\dot\gamma,\dot\gamma)=-1$.
In the next proposition, we fix $\bar r\neq 0$ and construct constants of motion in such a way $R$ has $\bar r$ as a root;  then we use Lemma~\ref{proposition flyby reaching the time-machine} to recover uniqueness of the root.
Let us recall  that if $\bar r\neq0$ belongs to the
$r$-section of equatorial time-machine region, then
$G(\bar r)<0$, where $G$ is the radial polynomial in \eqref{G}.
Equivalently, $\bar r^2(\bar r^2+a^2)<a^2(\bh^2-2M\bar r)$. In particular,
$\bh^2-2M\bar r>0$, and therefore
$\Delta(\bar r):=\bar r^2+a^2+\bh^2-2M\bar r>0$.
Thus, the square root $\sqrt{\Delta(\bar r)}$ which appears below is
well-defined.
\begin{thm}\label{prescribed_common_turning_fyby}
	Let $\bar r\neq0$ belong to the $r$-section of the equatorial time-machine
	region and set $\sigma:=\sgn\bar r$. For $\kappa>0$ and $\eta>0$, let
	\[
	k_\kappa:=-\frac{\kappa}{\bh},
	\qquad
	D_{\eta,\kappa}:=\sigma\eta\kappa,
	\]
	\begin{equation}\label{prescribed turning energy}
		E_{\eta,\kappa}:=
		\frac{-\kappa\bar r+aD_{\eta,\kappa}
			+\sqrt{\Delta(\bar r)}\sqrt{\bar r^2+D_{\eta,\kappa}^2}}
		{\bar r^2},
		\qquad
		L_{\eta,\kappa}:=aE_{\eta,\kappa}+D_{\eta,\kappa},
	\end{equation}
	and let $A_{\eta,\kappa}$ denote the quantity $A$ in \eqref{AABB} evaluated
	on the data $(k_\kappa,E_{\eta,\kappa},L_{\eta,\kappa})$. Assume that
	\begin{itemize}
		\item[(i)] $\eta<\dfrac aM$ \ and \ $E_{\eta,\kappa}\geq1$,
		 if $\sigma=+1$;
		\item[(ii)] $E_{\eta,\kappa}>1$ \ and \ $A_{\eta,\kappa}<0$,
		 if $\sigma=-1$.
	\end{itemize}
	Then there exists an equatorial flyby electromagnetic orbit in
	$\Omega_\sigma$ with radial turning value $\bar r$ and data
	$(k_\kappa,E_{\eta,\kappa},L_{\eta,\kappa},\mathcal Q=0)$. Moreover:
	\begin{itemize}
		\item[(a)] for every fixed $\eta>0$ (with $\eta<a/M$ if $\sigma=+1$)
		there exists $\kappa_0=\kappa_0(\eta)>0$ such that (i) or (ii)
		holds for every $\kappa\in(0,\kappa_0)$, respectively; in particular there exist infinitely many equatorial	flyby electromagnetic orbits, contained in $\Omega_\sigma$, having radial turning value  $\bar r$, corresponding to the two-parameters family of  
		data $(k_\kappa, E_{\eta, \kappa},L_{\eta,\kappa},\mathcal Q=0)$ defined by the map 
		$(\eta,\kappa)\mapsto(k_\kappa,E_{\eta,\kappa},L_{\eta,\kappa})$ which is 		smooth and injective on its  domain;
		\item[(b)] conversely, if $\gamma$ is a maximal equatorial,
		future-pointing electromagnetic orbit with data
		$(k,E,L,\mathcal Q=0)$ such that
		\[
		k\bh<0,
		\qquad
		\sigma(L-aE)>0,
		\]
		and $\dot\gamma^r=0$ at a point with radial coordinate $\bar r$, then,
		setting
		\[
		\kappa:=-k\bh,
		\qquad
		\eta:=\frac{\sigma(L-aE)}{\kappa},
		\]
		one has $E=E_{\eta,\kappa}$ and $L=L_{\eta,\kappa}$. Hence, if (i) or (ii) is satisfied respectively, $\gamma$ is a flyby electromagnetic orbit in
		$\Omega_\sigma$ with radial turning value $\bar r$.
	\end{itemize}
\end{thm}
\begin{proof}
	With $D=L-aE$, $R$ in \eqref{R} can be written as
	\begin{equation}\label{RP}
	R(r)=P(r)^2-\Delta(r)(r^2+D^2),
	\qquad
	P(r)=Er^2-aD-k\bh r.
	\end{equation}
	For the data $(k_\kappa,E_{\eta,\kappa},L_{\eta,\kappa})$, since
	$k_\kappa\bh=-\kappa$, a direct substitution gives
	\begin{equation}\label{PR}
		P(\bar r)=
		\sqrt{\Delta(\bar r)}\sqrt{\bar r^2+D_{\eta,\kappa}^2}>0
		\qquad\text{and}\qquad
		R(\bar r)=0.
	\end{equation}
	The coefficient $B$ in \eqref{AABB} is
	\[
	B=2D_{\eta,\kappa}\bigl(MD_{\eta,\kappa}+ak_\kappa\bh\bigr)
	=2\sigma\eta\kappa\bigl(M\sigma\eta\kappa-a\kappa\bigr)
	=2\eta\kappa^2\,(M\eta-\sigma a),
	\]
	so that
	\[
	\sigma B=2\eta\kappa^2\,(\sigma M\eta-a)<0
	\]
	if and only if $\eta<a/M$ when $\sigma=+1$, while it holds for every
	$\eta>0$ when $\sigma=-1$. Since moreover $\sgn k_\kappa=-\sgn\bh$ and, by
	assumption (i), resp.\ (ii), $E_{\eta,\kappa}\geq1$ when $\sigma=+1$, resp.\
	$E_{\eta,\kappa}>1$ and $A_{\eta,\kappa}<0$ when $\sigma=-1$, the hypotheses
	of Lemma~\ref{proposition flyby reaching the time-machine} are satisfied.
	Therefore $R$ has exactly one simple zero in $\Omega_\sigma$; since
	$R(\bar r)=0$ and $\bar r\in\Omega_\sigma$, this zero is precisely $\bar r$.
	The existence of the required equatorial flyby electromagnetic orbit then
	follows from Remark~\ref{remark flyby existence}, recalling that, from
	\eqref{PD}, $P(\bar r)=-g(u,V)$, where $u$ is the initial velocity vector,
	and that $P(\bar r)>0$ by \eqref{PR}.
	
	Let us now prove the statements (a) and (b).
	
	\noindent (a) We have
	\[
	\lim_{\kappa\to0^+}E_{\eta,\kappa}
	=\frac{\sqrt{\Delta(\bar r)}}{|\bar r|}>1,
	\]
	where the last inequality follows from
	$\Delta(\bar r)-\bar r^2=a^2+\bh^2-2M\bar r>0$: this is clear if
	$\bar r<0$, while if $\bar r>0$ it follows from $G(\bar r)<0$, which gives
	$a^2(\bh^2-2M\bar r)>\bar r^2(\bar r^2+a^2)>0$. Moreover,
	\[
	A_{\eta,\kappa}
	=-a^2-\bh^2-2aE_{\eta,\kappa}D_{\eta,\kappa}
	-D_{\eta,\kappa}^2+\bh^2k_\kappa^2
	\longrightarrow -a^2-\bh^2<0,
	\qquad \kappa\to0^+.
	\]
	Since $E_{\eta,\kappa}$ and $A_{\eta,\kappa}$ are continuous in $\kappa$,
	conditions (i), resp.\ (ii), hold for every $\kappa>0$ small enough. The map
	$(\eta,\kappa)\mapsto(k_\kappa,E_{\eta,\kappa},L_{\eta,\kappa})$ is clearly
	smooth; it is injective because $k_\kappa$ determines $\kappa$ and then
	$L_{\eta,\kappa}-aE_{\eta,\kappa}=\sigma\eta\kappa$ determines $\eta$.
	
	(b) Since $\gamma$ is equatorial with $\mathcal Q=0$ and $\dot\gamma^r=0$ at
	the point with radial coordinate $\bar r$, the third equation in
	\eqref{equations-of-motions} gives $R(\bar r)=0$, i.e.
	\[
	P(\bar r)^2=\Delta(\bar r)\bigl(\bar r^2+D^2\bigr).
	\]
	As $\gamma$ is future-pointing, $P(\bar r)=-g(\dot\gamma,V)>0$, hence
	$P(\bar r)=\sqrt{\Delta(\bar r)}\sqrt{\bar r^2+D^2}$. Solving the second equation in \eqref{RP} 	for $E$, with this value of $P(\bar r)$, $\kappa=-k\bh>0$ and $D=\sigma\eta\kappa$, yields exactly
	$E=E_{\eta,\kappa}$ in \eqref{prescribed turning energy}, and consequently
	$L=aE+D=L_{\eta,\kappa}$. If (i), resp.\ (ii), is satisfied, then, as in the
	first part of the proof, the hypotheses of
	Lemma~\ref{proposition flyby reaching the time-machine} hold; since $\gamma$
	is maximal and lies in $\Omega_\sigma$ (recall Remark~\ref{insideOmegasigma}) Lemma~\ref{proposition flyby reaching the time-machine} gives that $\gamma$ is a
	flyby electromagnetic orbit with radial turning value $\bar r$.
\end{proof}

\begin{oss}
	For the geometric construction of a broken electromagnetic orbit, the condition
	that the flyby branches have radial  turning value exactly equal to the radius
	$\bar r$ of the equatorial circular orbit is stronger than strictly necessary. It would
	be enough to require that the radial range of each flyby branch contain
	$\bar r$. Equivalently, if $\sigma=\sgn\bar r$, one could ask only that
	$\sigma r_{\mathrm{turn}}\leq \sigma\bar r$,
	and then choose the initial point of the corresponding flyby branch at
	$r=\bar r$. 
	We impose the stronger condition
	$r_{\mathrm{turn}}=\bar r$
	because, as observed in \cite [p. 69]{Ruffini_book},     it makes the decay vertices
	kinematically simpler. At such vertices the flyby branches and the circular
	orbit all have vanishing radial velocity. Thus the radial component of the
	kinetic four-momentum conservation law reduces to the contribution of the remaining decay
	products (see Remark~\ref{decay_turning_points}). 
\end{oss}
\section{Broken electromagnetic orbits and charged-particle decay processes}\label{decay_section}
A decay process in Kerr--Newman spacetime involving charged particles can be  described using multiple broken electromagnetic orbits:
\begin{defn}
A {\em broken electromagnetic orbit} is   a continuous curve which is a concatenation of electromagnetic orbits  not necessarily with the same charge-to-mass ratio $k$. 

\noindent A broken electromagnetic orbit is {\em closed} if the continuous curve is a closed one (see Figure~\ref{closed}).
\end{defn}
\begin{figure}[H]
\centering
\includegraphics[scale=0.4]{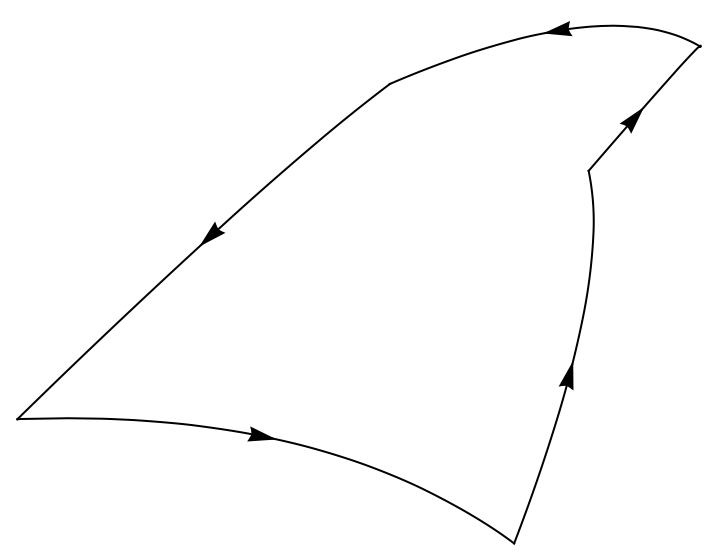} 
\vspace{0.2cm}
\caption{Closed broken electromagnetic orbit.}\label{closed}
\end{figure}
At a break point only continuity of the worldline is imposed; no matching of the one-sided proper-time derivatives is required. If the charge-to-mass ratio changes at the break point, the adjacent branches solve different Lorentz-force equations, so their one-sided accelerations are in general different. When the break point is interpreted as a decay vertex, the matching conditions are instead conservation of kinetic four-momentum and electric charge.

In the following, all the electromagnetic orbits lie on the equatorial plane $\{\theta=\pi/2\}$.\\
Let $\gamma_1$ be a massive particle of rest mass $m_{\gamma_1}$ departing from a point $p$ and whose motion is a flyby electromagnetic orbit (see Definition~\ref{orbits}), with charge-to-mass ratio $k_{\gamma_1}$ and such that its turning point  is contained in $\mathfrak{T}$. The latter condition can be satisfied as seen in Theorem~\ref{prescribed_common_turning_fyby}.

Suppose that at point $q\in\mathfrak T$ with $r(q)=\bar{r}$, the particle $\gamma_1$ decays into $n+1$ particles
\begin{align}\label{first decay}
	\gamma_1 \longrightarrow \alpha + \tilde{\gamma}_1 + \ldots +\tilde{\gamma}_n .
\end{align}
If the produced particle, say $\alpha$, happens to have a charge-to-mass ratio $k_{\alpha}$ as in Eq.~\eqref{kbarr} and energy and angular momentum per unit mass as in Eqs.~\eqref{energy middle curve},~\eqref{ang momentum middle curve}, the electromagnetic orbit $\alpha$ can be an equatorial circular orbit (precisely, a solution of the LFE as in Theorem~\ref{proposition alpha curve}) of radius $\bar{r}$. Now suppose that the particle $\alpha$ during its travel intersects some points $w\in I^{-}(\gamma_1)$
and assume that at such a point $\alpha$  decays into $m+1$ particles
\begin{align}\label{second decay}
	\alpha \longrightarrow \gamma_2 + \tilde{\alpha}_1 + \ldots + \tilde{\alpha}_m .
\end{align}
Since the charge-to-mass ratio of the particle, say $\gamma_2$, will differ from the one of $\alpha$, we choose the outgoing fragment $\gamma_2$ to have charge-to-mass ratio different from $k_\alpha$ so that it leaves the spherical branch as an equatorial electromagnetic orbit that might reach the initial point $p$ or another point on $\mathrm{Im}\,\gamma_1$.
\begin{oss}\label{another closing point}
	If $p\in \mathrm{I}\cup \mathrm{II}$, the curve $\gamma_2$ cannot return to the point $p$ itself; in this case the closing point must be some  $p'\in \big(\mathrm{Im}\,\gamma_1\cap I^+(w)\big)\cap \mathrm{III}$  which then we use to  replace $p$. On the other hand, if $\gamma_2$ is a flyby it might cross the inner horizon in outgoing coordinates and, after traversing a block II', reach another block I  in an extension of Kerr--Newman spacetime (see Remark~\ref{extensions}). The same might happen to  $\tilde \gamma_1$ produced in the first decay, provided it is a flyby electromagnetic orbit. This leads to a possible interpretation of the first decay as  a Penrose-type process between different  asymptotically flat ends, see Subsection~\ref{Penrose}.
\end{oss}
One closed broken electromagnetic orbit in the above process is then  given by the concatenation (see Figure~\ref{broken orbit figure})
\begin{align}\label{concatenation expression}
	\gamma_1 \cup \alpha \cup \gamma_2
\end{align}
(possibly by replacing $\gamma_1$ with one of its restrictions according to Remark~\ref{another closing point}).

In the decay processes \eqref{first decay} and \eqref{second decay}, the fundamental conservation law at the vertices $q,w\in\mathfrak{T}$ is that of the kinetic four-momentum $\pi=mu$, where $u$ is the timelike four-velocity and $m$ is the rest mass of the particle. Since the constants of motion $E$ and $L$ defined in \eqref{EL} are the energy and angular momentum per unit mass -- consistent with the parametrization $g(\dot\gamma,\dot\gamma)=-1$ and with the form \eqref{LFE-abstract} of the Lorentz force equation, in which the charge-to-mass ratio $k$ appears as the coupling constant -- the components of the kinetic momentum along the Killing fields $\partial_t$ and $\partial_\phi$ read
\[
-g(\pi,\partial_t)=mE+eA_t,
\qquad
g(\pi,\partial_\phi)=mL-eA_\phi,
\]
where $e=km$ is the electric charge of the particle. The conservation of $-g(\pi,\partial_t)$ and $g(\pi,\partial_\phi)$ at each decay vertex, together with the conservation of the total electric charge (which makes the $A_t$ and $A_\phi$ contributions cancel between the incoming and outgoing states), is therefore equivalent to the \emph{mass-weighted} conservation of $E$ and $L$. Accordingly, as $g(\pi,\partial_r)=m\,g_{rr}(\bar{r},\pi/2)\,u^r$,
at the decay point $q$ we must have
\begin{align*}
	e_{\gamma_1}&=e_{\alpha} + e_{\tilde{\gamma}_1} + \ldots + e_{\tilde{\gamma}_n},\\ m_{\gamma_1}E(\gamma_1)&=m_{\alpha} E(\alpha) + m_{\tilde{\gamma}_1} E(\tilde{\gamma}_1) + \ldots + m_{\tilde{\gamma}_n} E(\tilde{\gamma}_n), \\
	m_{\gamma_1} L(\gamma_1)&=m_{\alpha} L(\alpha) + m_{\tilde{\gamma}_1} L(\tilde{\gamma}_1) + \ldots + m_{\tilde{\gamma}_n} L(\tilde{\gamma}_n),\\
    m_{\gamma_1} \dot{\gamma}_1^r&=m_{\alpha} \dot{\alpha}^r + m_{\tilde{\gamma}_1} \dot{\tilde{\gamma}}^r_1 + \ldots + m_{\tilde{\gamma}_n} \dot{\tilde{\gamma}}^r_n,
\end{align*}
and analogously at the decay point $w$,
\begin{align*}
   e_{\alpha}&=e_{\gamma_2} + e_{\tilde{\alpha}_1} + \ldots + e_{\tilde{\alpha}_m}\\
   	m_{\alpha} E(\alpha)&=m_{\gamma_2} E(\gamma_2) + m_{\tilde{\alpha}_1} E(\tilde{\alpha}_1) + \ldots + m_{\tilde{\alpha}_m} E(\tilde{\alpha}_m),\\
	m_{\alpha} L(\alpha)&=m_{\gamma_2} L(\gamma_2) + m_{\tilde{\alpha}_1} L(\tilde{\alpha}_1) + \ldots + m_{\tilde{\alpha}_m} L(\tilde{\alpha}_m),\\
    m_{\alpha} \dot{\alpha}^r&=m_{\gamma_2} \dot{\gamma}_2^r + m_{\tilde{\alpha}_1} \dot{\tilde{\alpha}}^r_1 + \ldots + m_{\tilde{\alpha}_m} \dot{\tilde{\alpha}}^r_m.
\end{align*}
\begin{oss}\label{masses}  
The rest masses $m_{\gamma_1},m_{\alpha},m_{\gamma_2},m_{\tilde{\gamma}_i},m_{\tilde{\alpha}_j}$ are  subject to the kinematic admissibility conditions that the system of conservation laws admits a solution with all kinetic four-momenta $\pi_\bullet$ which have to be future-pointing timelike vectors (in particular, this entails the necessary inequalities $m_{\gamma_1}> m_{\alpha}+\sum_i m_{\tilde{\gamma}_i}$ and $m_{\alpha}> m_{\gamma_2}+\sum_j m_{\tilde{\alpha}_j}$, see \cite{landau_book}).
\end{oss}
\begin{oss}\label{decay_turning_points}
    Notice that if we require  the orbits $\gamma_1,\;\gamma_2$ to be flyby, $\alpha$ to be spherical, and the decays to happen at the turning points of $\gamma_1$ and $\gamma_2$, the conservation laws for $r$-component of $\pi$ reduce to
    \begin{align*}
       0 &= m_{\tilde{\gamma}_1} \dot{\tilde{\gamma}}^r_1 + \ldots + m_{\tilde{\gamma}_n} \dot{\tilde{\gamma}}^r_n,\\
        0 &= m_{\tilde{\alpha}_1} \dot{\tilde{\alpha}}^r_1 + \ldots + m_{\tilde{\alpha}_m} \dot{\tilde{\alpha}}^r_m.
    \end{align*}
\end{oss}
\subsection{Penrose-type processes}\label{Penrose}
The particle decays described above can be associated with  energy extraction.
More precisely, consider  
\[\gamma_1\longrightarrow \alpha + \tilde{\gamma}_1,\] 
happening in the region $\mathfrak{T}\cap\{r>0\}\cap\{\theta=\pi/2\}$.
By Remark \ref{remark sign spherical charge}, $\sgn{k_{\alpha}}=-\sgn{\bh}$. Therefore, despite $\partial_t$ being future-pointing timelike by Remark \ref{remark_partial_t_timelike_time_machine} and Proposition \ref{time_orientation_killing_vectors}, we may have
\begin{align} \label{E_alpha<0_condition}
    E(\alpha)=-g(\dot{\alpha},\partial_t) + \frac{k_{\alpha}\bh}{r}<0,
\end{align}
as for instance shown in Example~\ref{decay_positive}.
From Remark~\ref{masses} we  must have  $m_{\gamma_1}>m_{\alpha}+m_{\tilde{\gamma}_1}$. Hence if \eqref{E_alpha<0_condition} is satisfied, by the conservation of the mass-weighted energy we have
\begin{align} \label{energy_extracted_particle}
    E(\tilde{\gamma}_1)=\frac{m_{\gamma_1}E(\gamma_1)-m_{\alpha}E(\alpha)}{m_{\tilde{\gamma}_1}} > \frac{m_{\gamma_1} }{m_{\tilde{\gamma}_1}} E(\gamma_1) > E(\gamma_1).
\end{align}

\begin{figure}[H] 
\includegraphics[scale=0.25]{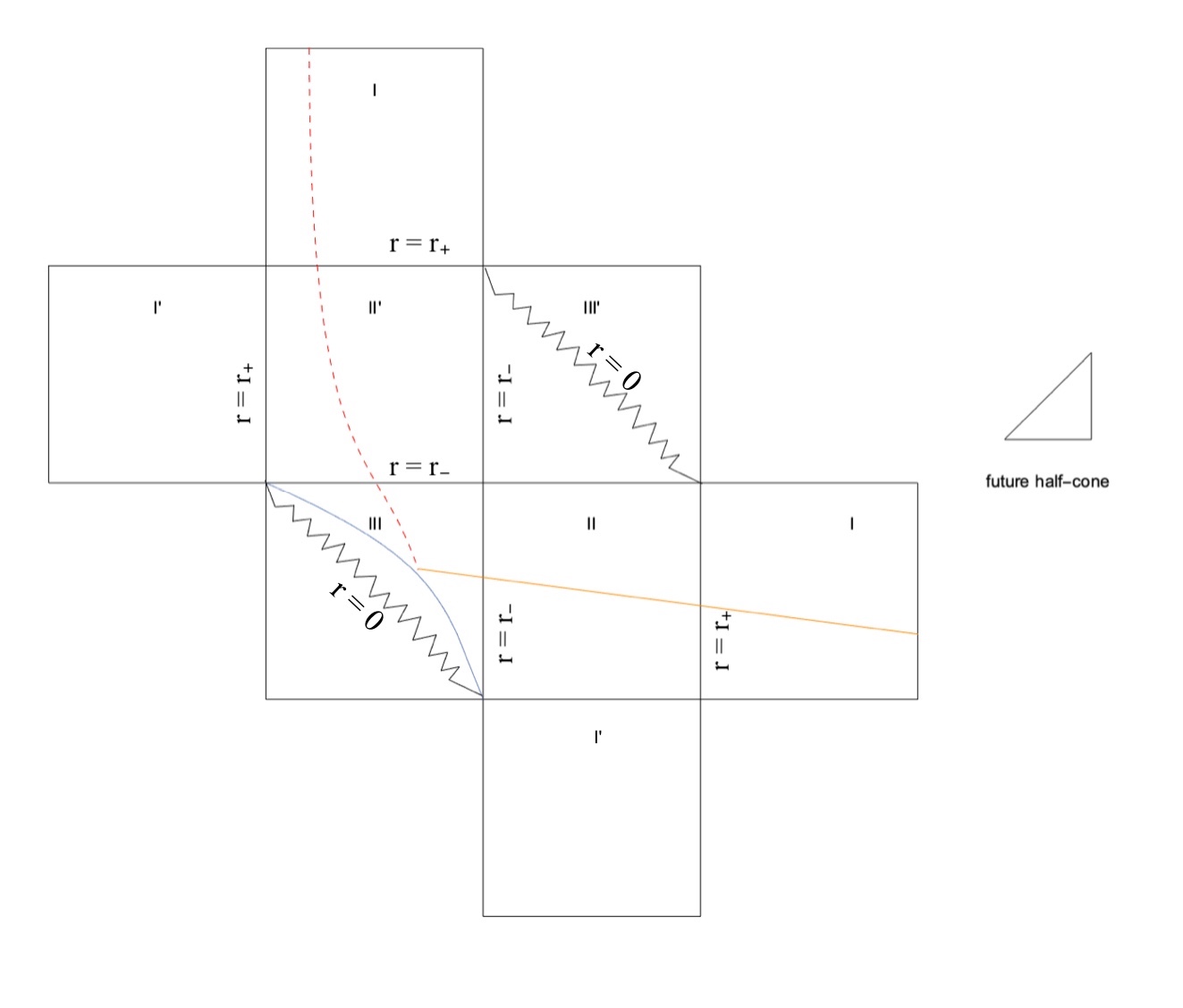}  
    \caption{
    A subset of the conformal diagram of the maximal analytic extension of the Kerr--Newman spacetime: flyby orbit $\gamma_1$ (orange curve) reaches a point of $\mathfrak{T}\cap\{r>0\}\subset\,$III where it decays in the equatorial circular orbit $\alpha$ (blue curve) and in the flyby orbit $\tilde{\gamma}_1$ (dashed red curve). The latter reaches the BL block I of another Kerr--Newman patch.} \label{figure_penrose_diagram}
\end{figure}

Equations \eqref{E_alpha<0_condition} and \eqref{energy_extracted_particle} allow us to interpret the process in a way analogous to the classical energy extraction Penrose process \cite{Penrose_process} happening in the ergosphere, i.e. the subset above the event horizon where $\partial_t$ is spacelike, but in the maximal analytic extension of Kerr--Newman spacetime, restricted to the equatorial plane (see for instance \cite[Fig.11.7]{Griffiths_Podolsky_2009} for the analogous conformal diagram of the Kerr spacetime).  Indeed, the decay process occurs  in block III, below the Cauchy horizon $\{r=r_{-}\}$, so $\tilde{\gamma}_1$ is not able to reach back the initial block I. However if the future-pointing particle $\tilde{\gamma}_1$ is a flyby orbit (as in  Example \ref{decay_positive}), it will pass through different horizons of a Kerr--Newman patch with opposite time-orientation, whose BL blocks are denoted as I',II',III', reaching a  block I of another Kerr--Newman patch. Hence, an observer living in the latter BL block I could gain energy from the outgoing particle $\tilde{\gamma}_1$, see Figure \ref{figure_penrose_diagram}.

An analogous interpretation is available in the negative-$r$ region. If the
decay $\gamma_1\to\alpha+\tilde\gamma_1$ occurs at a radius
$\bar r<0$ at which $E(\alpha)<0$ (as in Example~\ref{decay_negative}),
then \eqref{energy_extracted_particle} holds verbatim and 
 the extracted energy is carried to the
asymptotically flat end $r\to-\infty$ of block $\mathrm{III}$ itself, without
crossing any horizon.

\subsection{The causality violating process in block III}\label{subsection_causality_violation_process}
Let us now focus on the second decay
\[
\alpha\longrightarrow\gamma_2+\tilde\alpha_1.
\]
We show in Theorem~\ref{closed_broken_orbit_blockIII} that it can involve the production of a test particle $\gamma_2$
which, together with $\gamma_1$ and $\alpha$, forms a closed broken
electromagnetic orbit inside block $\mathrm{III}$.
We emphasize that Theorem~\ref{closed_broken_orbit_blockIII} is formulated independently of the
particular values of the constants of motion of the two flyby branches. The admissibility of the decay
vertices is a pointwise algebraic condition involving the kinetic four-momenta
and the electric charges. The
global closing argument below applies without restrictions on the
corresponding constants of motion, apart from $E>1$ in the negative $r$-region (see Remark~\ref{closing}).

Fix a nonzero radius $\bar r$ in the equatorial time-machine region and let
$\gamma_1$ and $\gamma_2$ be two equatorial flyby
electromagnetic orbits with the same radial turning value $\bar r$ and, if $\bar r<0$, with energies $E_1,E_2>1$. Such orbits are provided by
Theorem~\ref{prescribed_common_turning_fyby}. 
Then, for $i=1,2$,
\[
R_i(\bar r)=0,\qquad R_i'(\bar r)\neq0,\qquad P_i(\bar r)>0.
\]
We then take $\alpha$ to be the equatorial circular electromagnetic
orbit of radius $\bar r$ given by Theorem~\ref{proposition alpha curve} (recall also Remark~\ref{elixorbit}).
The constants of motion and the
charge-to-mass ratio of $\alpha$ are then fixed by $\bar r$, and they are not
required to coincide with those of the flyby branches.

We choose
the proper-time parameters so that the turning point corresponds to
$\tau=0$. Let $k_i$ and $(E_i,L_i,\mathcal Q_i=0)$, $i=1,2$, be respectively their charge-to-mass ratio
and their constants of motion. Recalling \eqref{PD}, let us set
\begin{equation*}
P_i(r):=E_i(r^2+a^2)-aL_i-k_i\bh r,
\qquad
D_i:=L_i-aE_i,
\end{equation*}
and
\begin{equation}\label{AB}
\mathcal A_i(r):=
aD_i+\frac{(r^2+a^2)P_i(r)}{\Delta(r)},
\qquad
\mathcal B_i(r):=
D_i+\frac{aP_i(r)}{\Delta(r)}.
\end{equation}
We shall use $r$ as parameter on each monotone  branch of the flyby electromagnetic orbit $\gamma_i$. 
Indeed, away
from the turning point, the first-order equations \eqref{equations-of-motions} on the equatorial plane reduce to 
\[
r^2\dot t=\mathcal A_i(r),
\qquad
r^2\dot\phi=\mathcal B_i(r),
\qquad
r^2\dot r=\sigma\sqrt{R_i(r)},
\]
where $\sigma=1$ on a branch on which $r$ increases and
$\sigma=-1$ on a branch on which $r$ decreases. Hence
\[
\frac{dt}{dr}
=
\sigma\frac{\mathcal A_i(r)}{\sqrt{R_i(r)}},
\qquad
\frac{d\phi}{dr}
=
\sigma\frac{\mathcal B_i(r)}{\sqrt{R_i(r)}}.
\]
We distinguish two cases.
\smallskip

\noindent\emph{Negative-$r$.} In this case,  $\bar r<0$, and since $\gamma_i$ is a flyby orbit with turning point $\bar r$:
\[
R_i(\bar r)=0,\qquad
R_i'(\bar r)<0,\qquad
R_i(r)>0\quad\text{for every }r<\bar r.
\]
In this case the  branch issued from the turning point and escaping
towards $r=-\infty$ has $r$ decreasing. Therefore, for $r<\bar r$, the
coordinate lapses from $\bar r$ to $r$ are
\begin{align}
T_i^-(r)&:=\int_{r}^{\bar r}\frac{\mathcal A_i(s)}{\sqrt{R_i(s)}}\,ds, \label{negative_lapse_Ti} \\
\Phi_i^-(r)&:=\int_{r}^{\bar r}\frac{\mathcal B_i(s)}{\sqrt{R_i(s)}}\,ds.\label{negative_lapse_Phi_i} 
\end{align}
The incoming branch reaching the turning point from the region $r<\bar r$
has $r$ increasing. Hence the coordinate lapse from the point with radial
coordinate $r$ to the turning point is again $T_i^-(r)$, respectively
$\Phi_i^-(r)$.

\smallskip

\noindent\emph{Positive-$r$.} In this case $0<\bar r<r_{\mathrm{tm}}$,
and 
\[
R_i(\bar r)=0,\qquad
R_i'(\bar r)>0,\qquad
R_i(r)>0\quad\text{for every }r\in(\bar r,r_-).
\]
Thus the future branch issued from the turning point and contained in
$(\bar r,r_-)$ has $r$ increasing. Therefore, for $r\in(\bar r,r_-)$, the
coordinate lapses from $\bar r$ to $r$ are
\begin{align}
T_i^+(r)&:=	\int_{\bar r}^{r}	\frac{\mathcal A_i(s)}{\sqrt{R_i(s)}}\,ds,	\label{positive_lapse_Ti}\\
\Phi_i^+(r)	&:=	\int_{\bar r}^{r}	\frac{\mathcal B_i(s)}{\sqrt{R_i(s)}}\,ds.
\label{positive_lapse_Phi_i}
\end{align}
The incoming branch reaching the turning point has $r$ decreasing. Hence, the
coordinate lapse from the point with radial coordinate $r$ to the turning
point is again $T_i^+(r)$, respectively $\Phi_i^+(r)$.

\smallskip
In what follows we write $T_i$ and $\Phi_i$ for $T_i^+$ and $\Phi_i^+$ in
the positive-$r$ case, and for $T_i^-$ and $\Phi_i^-$ in the negative-$r$  case.

\begin{thm}
    \label{closed_broken_orbit_blockIII}
Let $\gamma_1$ and $\gamma_2$ be equatorial flyby
electromagnetic orbits with the same radial turning value $\bar r$ and, if $\bar r<0$, with energies $E_1,E_2>1$, and   let $\alpha$  be an equatorial circular electromagnetic 
orbit of radius $\bar r$ according to Theorem~\ref{proposition alpha curve}. In the positive-$r$ case there exists a sequence $r_n\to r_-^-$, while in
the negative-$r$ case there exists a sequence $r_n\to-\infty$, such that
\begin{equation}\label{phasecondition}
\Phi_1(r_n)+\Phi_2(r_n)
+\omega\bigl(T_1(r_n)+T_2(r_n)\bigr)
\in 2\pi\mathbb Z.
\end{equation}
Moreover,
\[
s_n:=T_1(r_n)+T_2(r_n)>0,
\]
for all $n$ large enough, and  setting $w_n:=\alpha(s_n)$,
the branch of the flyby electromagnetic orbit $\gamma_2$ issued from
$w_n$ meets the incoming branch of $\gamma_1$ at a point $p_n$ with
$r(p_n)=r_n$. Hence, for suitable parameters $\tau_n<0$ and $\sigma_n>0$,
\begin{equation}\label{chain}
\gamma_1\big|_{[\tau_n,0]}
\cup
\alpha\big|_{[0,s_n]}
\cup
\gamma_2\big|_{[0,\sigma_n]}
\end{equation}
is a closed broken electromagnetic orbit
contained in block $\mathrm{III}$.
\end{thm}
\begin{proof}
Let us first prove the fact that the broken electromagnetic orbit  \eqref{chain} is  a closed one. 
Let $r_*$ belong to the relevant
radial interval, namely $r_*\in(\bar r,r_-)$ in the positive-$r$ case and
$r_*<\bar r$ in the negative-$r$ case. Suppose that
\[
s_*:=T_1(r_*)+T_2(r_*)>0
\]
and
\[
\Phi_1(r_*)+\Phi_2(r_*)
+\omega\bigl(T_1(r_*)+T_2(r_*)\bigr)
\in 2\pi\mathbb Z.
\]
By the previous discussion on the coordinate lapses and recalling that $q=\alpha(0)=(t_0,\bar r,\pi/2,\phi_0)$,  the point of the incoming branch of $\gamma_1$ with radial coordinate $r_*$ has coordinates
\[
p_*=\bigl(t_0-T_1(r_*),	r_*,	\pi/2,	\phi_0-\Phi_1(r_*)\bigr).
\]
Indeed, $T_1(r_*)$ and $\Phi_1(r_*)$ are precisely the $t$-- and $\phi$--coordinate lapses from $p_*$ to the turning point $q$ along the incoming branch of $\gamma_1$.
On the other hand, if
\[
w=\alpha(s_*)=(t_0-s_*,\bar r,\pi/2,\phi_0+\omega s_*),
\]
then the  branch of $\gamma_2$ issued from $w$ reaches the radial	value $r_*$ at the point
\[\bigl(t_0-s_*+T_2(r_*),	r_*,	\pi/2,	\phi_0+\omega s_*+\Phi_2(r_*)	\bigr),	\]
because $T_2(r_*)$ and $\Phi_2(r_*)$ are the coordinate lapses from the
turning point $\bar r$ to the point with radial coordinate $r_*$ along that branch of $\gamma_2$.
This point coincides with $p_*$ if and only if
\[
t_0-s_*+T_2(r_*)=t_0-T_1(r_*),
\]
and
\[
\phi_0+\omega s_*+\Phi_2(r_*)\equiv\phi_0-\Phi_1(r_*),\qquad \bmod 2\pi.
\]
The first equation is precisely
\[
s_*=T_1(r_*)+T_2(r_*).
\]
Using this value of $s_*$, the angular equation becomes
\[
\Phi_1(r_*)+\Phi_2(r_*)	+\omega\bigl(T_1(r_*)+T_2(r_*)\bigr)\in 2\pi\mathbb Z.
\]
Thus, whenever \eqref{phasecondition} holds and $s_*>0$, the  branch of
$\gamma_2$ issued from $\alpha(s_*)$ meets the incoming branch of
$\gamma_1$ at $p_*$.
	
It remains to prove that a sequence of radii satisfying \eqref{phasecondition} exists. Let
\[F(r):=	\Phi_1(r)+\Phi_2(r)	+\omega\bigl(T_1(r)+T_2(r)\bigr).\]
In the positive-$r$ case, by \eqref{AB} and \eqref{positive_lapse_Ti}--\eqref{positive_lapse_Phi_i}, we have
\[
F(r)=	\sum_{i=1}^{2}	\int_{\bar r}^{r}	\left(	\frac{(1+a\omega)D_i}{\sqrt{R_i(s)}}
+\frac{\bigl(a+\omega(s^2+a^2)\bigr)P_i(s)}	{\Delta(s)\sqrt{R_i(s)}}\right)ds.
\]
Since $\gamma_i$ is future-pointing and
$P_i(r)=-g(\dot\gamma_i,V)$, where $V$ is the future-directed timelike field orienting block
$\mathrm{III}$, we have $P_i(r)>0$ for $r\in(\bar r,r_-)$. Moreover, $V$ extends to $r=r_-$ as a nonzero future-pointing causal vector field (recall Remark~\ref{horizons}), so that $P_i(r_-)>0$.
Therefore, from \eqref{R}, 
\[
R_i(r_-)=P_i(r_-)^2>0.
\]
The first integrals in the expression of $F$ converge as $r\to r_-^-$. The
second ones contain the singular term $\Delta(r)=(r-r_-)(r-r_+)$; by \eqref{1+aomega},
\[
a+\omega(r_-^2+a^2)	=	a(1+a\omega)+\omega r_-^2<0,\]
hence they both  diverge negatively as $r\to r_-^-$.  Therefore, 
$F(r)\to-\infty$, as $r\to r_-^-$ and then it 
crosses infinitely many levels $2\pi n$ with $n\to-\infty$.
This gives a sequence $r_n\to r_-^-$ such that $F(r_n)\in2\pi\mathbb Z$.
Analogously, from \eqref{positive_lapse_Ti} and the same horizon asymptotics,
$T_i(r)\to+\infty$, for each $i=1,2$,
as $r\to r_-^-$. Thus, $s_n:=T_1(r_n)+T_2(r_n)>0$ for all $n$ sufficiently large.

In the negative-$r$ case, from \eqref{Rr} we have, as $r\to-\infty$,
\[
\sqrt{R_i(r)}=	\sqrt{E_i^2-1}\,r^2+O(|r|).	\]
Moreover,
\[
\frac{\mathcal A_i(r)}{\sqrt{R_i(r)}}=	\frac{aD_i+(r^2+a^2)P_i(r)/\Delta(r)}{\sqrt{R_i(r)}}
=\frac{E_i}{\sqrt{E_i^2-1}}+O(|r|^{-1}),	\]
hence by \eqref{negative_lapse_Ti}, $T_i(r)	\to+\infty$, as $r\to-\infty$.
On the other hand,
\[
\frac{\mathcal B_i(r)}{\sqrt{R_i(r)}}=	\frac{D_i+aP_i(r)/\Delta(r)}	{\sqrt{R_i(r)}}	=	O(r^{-2}),
\]
and therefore $\Phi_i(r)$ in \eqref{negative_lapse_Phi_i} converges  as $r\to-\infty$. Since the
future-pointing spherical orbit $\alpha$ satisfies $\omega<0$, we obtain
\[
F(r)	=	\omega\bigl(T_1(r)+T_2(r)\bigr)+O(1)	\to-\infty,	\qquad r\to-\infty.
\]
Therefore, $F$ crosses infinitely many levels $2\pi n$ with $n\to-\infty$.
This gives a sequence $r_n\to-\infty$ such that
$F(r_n)\in2\pi\mathbb Z$. Since also $T_i(r_n)\to+\infty$, $i=1,2$, we also have that
$s_n:=T_1(r_n)+T_2(r_n)>0$ for all $n$ large enough.	
\end{proof}
\begin{oss}\label{closing}
Theorem~\ref{closed_broken_orbit_blockIII} is not a result for specially tuned constants of motion.
Once the radius $\bar r$ in the equatorial time-machine region has been fixed,
the spherical branch $\alpha$ is fixed by Theorem~\ref{proposition alpha curve}, but the two flyby
branches may be chosen arbitrarily among the equatorial future-pointing flyby
electromagnetic orbits with radial turning value $\bar r$; in the negative
$r$-case we only require $E_1,E_2>1$. No further global matching condition is
imposed on their constants of motion or charge-to-mass ratios.
Thus the closing result is universal
within this class of flyby--spherical--flyby concatenations.
In particular, the interpretation  of the construction in Theorem~\ref{closed_broken_orbit_blockIII} as a  causality-violating particle decay is not confined to the
numerical examples of Section~\ref{subsection_numerics}. It only requires the local conservation laws at the break
points. These are algebraic conditions in one tangent space and are
generically solvable. Indeed, let an incoming particle have kinetic four-momentum
$p_0=m_0u_0$, and suppose that one outgoing branch is prescribed with
kinetic four-momentum $p_1=m_1u_1$, where $u_0,u_1$ are future-pointing unit
timelike vectors. The remaining fragment must then have kinetic four-momentum $p_2=p_0-p_1$.
Writing $-g(u_0,u_1)=\cosh \chi$, one finds
\[
g(p_2,p_2)	=	-m_0^2-m_1^2+2m_0m_1\cosh\chi .
\]
Hence $p_2$ is timelike for all sufficiently small positive ratios
$m_1/m_0$. Moreover $p_2$ is then future-pointing, and its rest mass is
$m_2=\sqrt{-g(p_2,p_2)}$.
The electric charge conservation law imposes no additional obstruction:
it simply determines
\[
e_2=e_0-e_1,
\qquad
k_2=\frac{e_2}{m_2}.
\]
Thus the vertex compatibility conditions are nonempty, open conditions.
\end{oss}
\begin{oss}\label{brokenKerr}
Notice that an analogous concatenation  cannot be achieved in the Kerr spacetime using future-pointing timelike  geodesics since  a timelike curve  with constant $r$-coordinate  on $\{\theta=\pi/2\}\cap\mathfrak{T}$ cannot be a geodesic, as shown in the following Proposition.
\end{oss}
\begin{prop}
A timelike curve $\alpha(s)=(t_0+b_t s,\bar{r},\pi/2,\phi_0+b_{\phi} s)$ in the Kerr spacetime, contained in $\mathfrak{T}$, with  $t_0,b_t,  b_\phi\in\mathbb{R},\,  \phi_0\in[0,2\pi)$, cannot be a geodesic.
\end{prop}
\begin{proof}
First,  observe that a curve $\alpha$ of this type would be the most general geodesic curve where the $r,\theta$-coordinates are kept constant by Theorem~\ref{first-order-equations}.
In the Kerr spacetime the time-machine region lies in $\{r<0\}$, so $\bar{r}<0$.  Since by Proposition~4.8.2 of \cite{KBH_book} every timelike geodesic, with Carter constant $\mathcal Q\geq 0$ and constrained in $\{r<0\}$, is a flyby orbit,  a geodesic with negative constant $r$-coordinate must have $\mathcal Q<0$,  consistently with Proposition $4.4$ of \cite{sanzeni_timelike}. However,  the constancy of $\theta=\pi/2$ implies that $\mathcal Q=0$ by the system \eqref{equations-of-motions}.
\end{proof}

\begin{figure}[H] 
	\centering
	\includegraphics[scale=0.35]{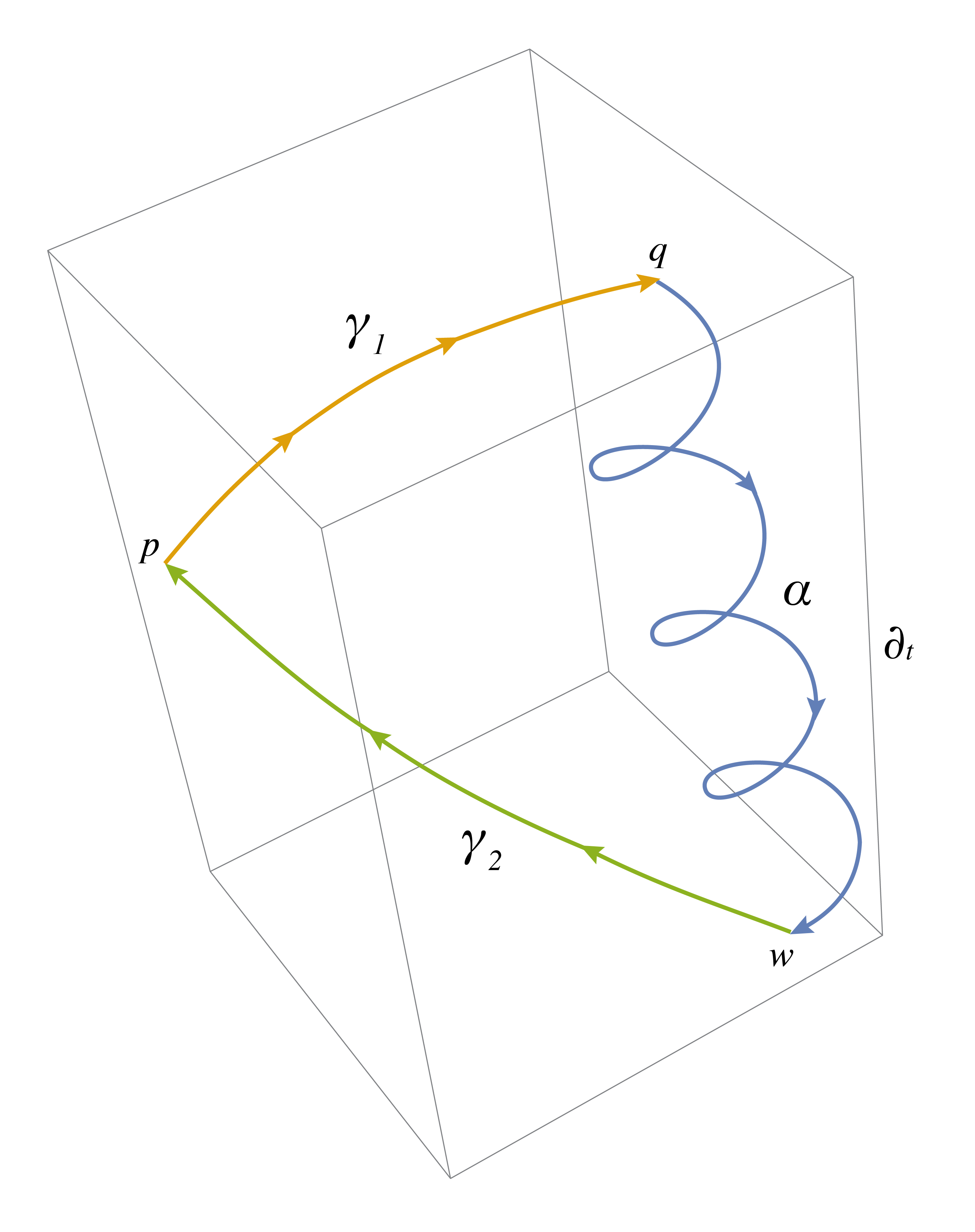}  
	\caption{The concatenation \eqref{concatenation expression} of solutions to the LFE: an illustrative picture of the closed broken electromagnetic orbit in Theorem~\ref{closed_broken_orbit_blockIII}} \label{broken orbit figure}
\end{figure}

\subsection{Some numerical examples} \label{subsection_numerics}
Let us now show some numerical examples of the decay processes described in Sections~\ref{Penrose}--\ref{subsection_causality_violation_process} in both $r$-positive and $r$-negative equatorial regions of block III.

\begin{example} \label{decay_positive}
	Let $a=3,\, M=5,\, \bh=3.5$, so that $r_{\mathrm{tm}}\approx 1.090$, and fix
	$\bar r=0.9\in(0,r_{\mathrm{tm}})$.
	
	\smallskip
	\noindent\emph{Step 1: the first flyby orbit.}
Consider the data
\[
m_{\gamma_1}=10,\qquad k_{\gamma_1}=-1,\qquad e_{\gamma_1}=-10,\qquad
E(\gamma_1)=1.
\]
We have
\[
\kappa=-k_{\gamma_1}\bh=3.5.
\]
Let $D_{\gamma_1}:=L(\gamma_1)-aE(\gamma_1)$;
replacing in \eqref{RP} and imposing $R_{\gamma_1}(\bar r)=0$, we get two
solutions, both with $P_{\gamma_1}(\bar r)>0$; we select the one with
$D_{\gamma_1}>0$, namely
\[
D_{\gamma_1}\approx 0.207,
\qquad
L(\gamma_1)\approx 3.207.
\]
Thus
\[
\eta:=\frac{D_{\gamma_1}}{\kappa}\approx 0.059<\frac{a}{M}=0.6,
\]
and, by Theorem~\ref{prescribed_common_turning_fyby}-(b),
$E_{\eta,\kappa}=E(\gamma_1)=1$ and
$L_{\eta,\kappa}=a+D_{\gamma_1}\approx 3.207$, so that condition (i) is
satisfied. Hence there exists an equatorial flyby electromagnetic orbit
$\gamma_1$ with prescribed turning value $\bar r=0.9$.

\smallskip
	\noindent\emph{Step 2: the equatorial circular orbit.}
	Let $\alpha$ be the equatorial circular electromagnetic orbit of radius $\bar r$ given
	by Theorem~\ref{proposition alpha curve} (recall
	Remark~\ref{elixorbit}), with rest mass $m_\alpha=0.72$. From
	\eqref{kbarr} and Corollary~\ref{cor energy_ang spherical orbit},
	\[
	k_\alpha\approx -2.771,\quad e_\alpha\approx-1.995,\quad
	E(\alpha)\approx -7.408,\quad L(\alpha)\approx -26.055,
	\]
	as in Figure~\ref{elicarpos}. In particular $E(\alpha)<0$, i.e.
	\eqref{E_alpha<0_condition} holds.
	
	\smallskip
	\noindent\emph{Step 3: the first decay
		$\gamma_1\to\alpha+\tilde\gamma_1$ (Fig.~\ref{firstdecay_positive})}
	Let the decay occur at the turning point $q=\alpha(0)$. By the conservation
	laws of electric charge, mass-weighted energy and angular momentum, and the
	normalization $g(\dot{\tilde\gamma}_1,\dot{\tilde\gamma}_1)=-1$,
	\begin{align*}
		&m_{\tilde{\gamma}_1}\approx 4.669,\\
		&k_{\tilde{\gamma}_1}= \frac{k_{\gamma_1} m_{\gamma_1}-k_{\alpha}m_{\alpha}}{m_{\tilde{\gamma}_1}}\approx -1.714,\\
		&E(\tilde{\gamma}_1)=\frac{m_{\gamma_1}E(\gamma_1)-m_{\alpha}E(\alpha)}{m_{\tilde{\gamma}_1}}\approx 3.284,\\
		&L(\tilde{\gamma}_1)=\frac{m_{\gamma_1}L(\gamma_1)-m_{\alpha}L(\alpha)}{m_{\tilde{\gamma}_1}}\approx 10.887,\\
		&e_{\tilde{\gamma}_1}=k_{\tilde{\gamma}_1} m_{\tilde{\gamma}_1}\approx -8.005,
	\end{align*}
	with $g(\dot{\tilde{\gamma}}_1,V)\approx -4.956<0$, so $\tilde\gamma_1$ is
	future-pointing, and $m_{\gamma_1}>m_\alpha+m_{\tilde\gamma_1}$. By
	Remark~\ref{decay_turning_points}, $\dot{\tilde\gamma}_1^r=0$ at $q$. The
	data of $\tilde\gamma_1$ satisfy the assumptions 
	in Theorem~\ref{prescribed_common_turning_fyby}-(b):
	\[
	\kappa=-k_{\tilde\gamma_1}\bh\approx 6.000,
	\qquad
	\eta=\frac{L(\tilde\gamma_1)-aE(\tilde\gamma_1)}{\kappa}\approx 0.172
	<\frac aM=0.6,
	\]
	and $E(\tilde\gamma_1)\approx 3.284\geq1$, so condition (i) is also satisfied.
	Hence $\tilde\gamma_1$ is a flyby electromagnetic orbit with the same
	turning value $\bar r$. Since
	$E(\alpha)<0$, inequality \eqref{energy_extracted_particle} applies and
	yields the energy gain
	\[
	E(\tilde\gamma_1)\approx 3.284>1=E(\gamma_1),
	\]
	realizing the Penrose-type process  of Subsection~\ref{Penrose} in the positive-$r$ region.
	
	\begin{figure}[H]
		\includegraphics[scale=0.6]{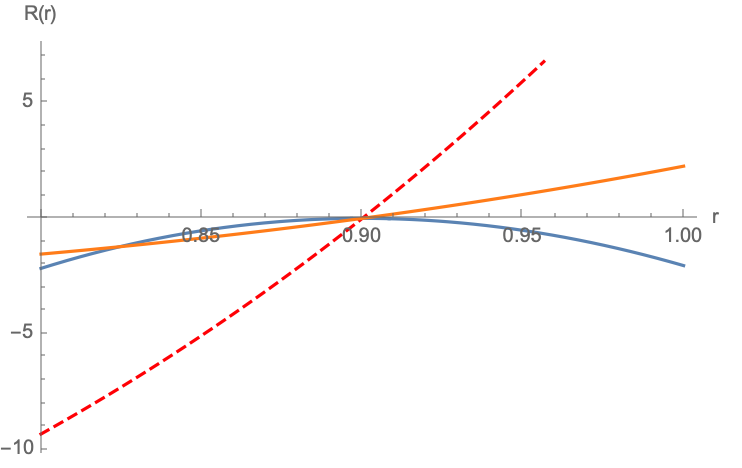}
		\caption{A common root of three polynomials $R$ for the first decay: a flyby electromagnetic orbit $\gamma_1$ (whose $R$ is the orange curve) reaching the radius $\bar{r}=0.9$ and decaying in the spherical orbit $\alpha$ in $\mathfrak T$ (whose polynomial $R$ is the blue curve) and in the flyby orbit $\tilde{\gamma}_1$ (whose $R$ is the dashed red curve).} \label{firstdecay_positive}
	\end{figure}
	
	\smallskip
	\noindent\emph{Step 4: the second flyby orbit and the second decay
		$\alpha\to\gamma_2+\tilde\alpha_1$ (Fig.~\ref{seconddecay_positive}).}
	Let
\[
\kappa=7,\ \text{i.e.}\ k_{\gamma_2}=-2,
\qquad
\eta\approx0.148,\ \text{i.e.}\ D_{\gamma_2}\approx 1.0385;
\]
condition (i) is satisfied, since $\eta<a/M$ and
\eqref{prescribed turning energy} gives $E(\gamma_2)=2.2\geq1$, with
$L(\gamma_2)\approx 7.638$. Thus, by
Theorem~\ref{prescribed_common_turning_fyby}, there exists another flyby
electromagnetic orbit $\gamma_2$ with turning value $\bar r$. Set
$m_{\gamma_2}=0.02$, hence $e_{\gamma_2}=-0.04$. Let the decay occur at a
point $w=\alpha(s)$, $s>0$, to be fixed in Step 5. By the conservation laws
and the normalization $g(\dot{\tilde\alpha}_1,\dot{\tilde\alpha}_1)=-1$,
\begin{align*}
	& m_{\tilde{\alpha}_1}\approx 0.430,\\
	&k_{\tilde{\alpha}_1}= \frac{k_{\alpha} m_{\alpha}-k_{\gamma_2}m_{\gamma_2}}{m_{\tilde{\alpha}_1}}\approx -4.544,\\
	&E(\tilde{\alpha}_1)=\frac{m_{\alpha}E(\alpha)-m_{\gamma_2}E(\gamma_2)}{m_{\tilde{\alpha}_1}}\approx -12.499,\\
	&L(\tilde{\alpha}_1)=\frac{m_{\alpha}L(\alpha)-m_{\gamma_2}L(\gamma_2)}{m_{\tilde{\alpha}_1}}\approx -43.952,\\
	&e_{\tilde{\alpha}_1}=k_{\tilde{\alpha}_1} m_{\tilde{\alpha}_1}\approx -1.955,
\end{align*}
with $g(\dot{\tilde{\alpha}}_1,V)\approx -23.558<0$ and
$m_\alpha>m_{\gamma_2}+m_{\tilde\alpha_1}$, so the decay is kinematically admissible.

	\begin{figure}[H]
		\includegraphics[scale=0.6]{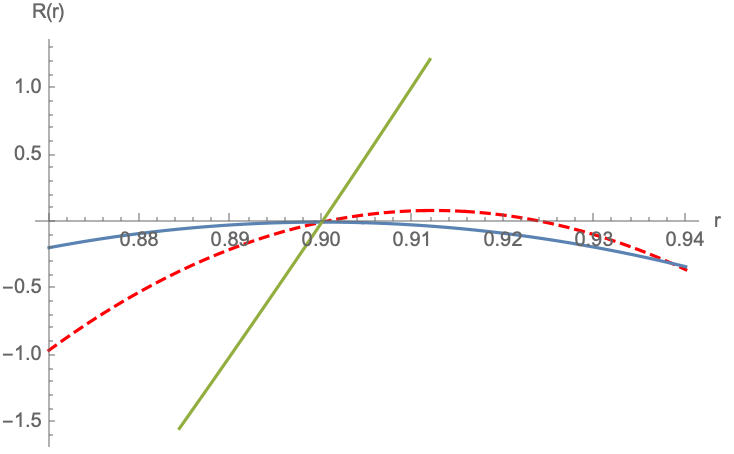}
		\caption{Other two polynomials $R$ having $\bar r=0.9$ as a common root for the second decay: a flyby orbit $\gamma_2$ (whose polynomial $R$ is the green curve) and an $r$-bouncing orbit $\tilde{\alpha}_1$ (whose polynomial is the dashed red curve). The polynomial $R$ of the spherical orbit $\alpha$ in Figure~\ref{firstdecay_positive} is also represented.}\label{seconddecay_positive}
	\end{figure}
	
	\smallskip
	\noindent\emph{Step 5: the closed broken electromagnetic orbit.}
	The orbits $\gamma_1$ and $\gamma_2$ are equatorial, future-pointing flyby
	electromagnetic orbits with the same radial turning value $\bar r=0.9$.
	By Theorem~\ref{closed_broken_orbit_blockIII} there exists a sequence
	$r_n\to r_-^-$ satisfying \eqref{phasecondition} with
	$s_n=T_1(r_n)+T_2(r_n)>0$ for all $n$ large enough. Choosing the second
	decay at the point $w_n=\alpha(s_n)$, the concatenation
	\[
	\gamma_1\big|_{[\tau_n,0]}\cup\alpha\big|_{[0,s_n]}\cup\gamma_2\big|_{[0,\sigma_n]}
	\]
	is a closed broken electromagnetic orbit contained in the positive-$r$ region of block~$\mathrm{III}$.
\end{example}

\begin{example} \label{decay_negative}
Let $a=3,\, M=5,\, \bh=1$ and prescribe the radius
	$\bar r=-2.5$. Let us observe that $\bar r$ belongs to the $r$-section of
	the equatorial time-machine region. Indeed, the latter is characterized by
	$G(r)<0$, $r\neq0$, where $G$ is the radial polynomial in \eqref{G}, which
	for these values of the parameters reads
	\[
	G(r)=r^4+9r^2+90r-9,
	\]
	and $G(-2.5)\approx-138.687<0$.

\smallskip    
\noindent\emph{Step 1: the first flyby orbit.}
We choose $\gamma_1$ in the family of
Theorem~\ref{prescribed_common_turning_fyby} with
\[
\kappa=1,\ \text{i.e.}\ k_{\gamma_1}=-1,
\qquad
\eta=0.4,\ \text{i.e.}\ D_{\gamma_1}=-0.4;
\]
then \eqref{prescribed turning energy} gives
\[
E(\gamma_1)\approx 2.810,
\qquad
L(\gamma_1)\approx 8.029,
\]
and $A_{\eta,\kappa}\approx -2.417<0$, so that condition (ii) is satisfied.
Hence there exists an equatorial flyby electromagnetic orbit $\gamma_1$ with
prescribed turning value $\bar r=-2.5$. We set $m_{\gamma_1}=7$, hence
$e_{\gamma_1}=-7$.

\smallskip
\noindent\emph{Step 2: the equatorial circular orbit.}
Let $\alpha$ be the equatorial circular electromagnetic orbit of radius $\bar r$ given
by Theorem~\ref{proposition alpha curve} (recall
Remark~\ref{elixorbit}), with rest mass $m_\alpha=0.2$. From
\eqref{kbarr} and Corollary~\ref{cor energy_ang spherical orbit},
\[
k_\alpha\approx 17.045,\quad e_\alpha\approx 3.409,\quad
E(\alpha)\approx -0.479,\quad L(\alpha)\approx -12.489.
\]
In particular $E(\alpha)<0$, i.e. \eqref{E_alpha<0_condition} holds also in
the negative-$r$ region.

\smallskip
\noindent\emph{Step 3: the first decay
	$\gamma_1\to\alpha+\tilde\gamma_1$.}
Let the decay occur at the turning point $q=\alpha(0)$. By the conservation
laws of electric charge, mass-weighted energy and angular momentum, and the
normalization $g(\dot{\tilde\gamma}_1,\dot{\tilde\gamma}_1)=-1$,
\begin{align*}
	&m_{\tilde{\gamma}_1}\approx 6.178,\\
	&k_{\tilde{\gamma}_1}= \frac{k_{\gamma_1} m_{\gamma_1}-k_{\alpha}m_{\alpha}}{m_{\tilde{\gamma}_1}}\approx -1.685,\\
	&E(\tilde{\gamma}_1)=\frac{m_{\gamma_1}E(\gamma_1)-m_{\alpha}E(\alpha)}{m_{\tilde{\gamma}_1}}\approx 3.199,\\
	&L(\tilde{\gamma}_1)=\frac{m_{\gamma_1}L(\gamma_1)-m_{\alpha}L(\alpha)}{m_{\tilde{\gamma}_1}}\approx 9.502,\\
	&e_{\tilde{\gamma}_1}=k_{\tilde{\gamma}_1} m_{\tilde{\gamma}_1}\approx -10.409,
\end {align*}
with $g(\dot{\tilde{\gamma}}_1,V)\approx -16.068<0$, so $\tilde\gamma_1$ is
future-pointing, and $m_{\gamma_1}>m_\alpha+m_{\tilde\gamma_1}$. By
Remark~\ref{decay_turning_points}, $\dot{\tilde\gamma}_1^r=0$ at $q$. The
data of $\tilde\gamma_1$ satisfy the assumptions in
Theorem~\ref{prescribed_common_turning_fyby}-(b):
\[
\kappa=-k_{\tilde\gamma_1}\bh\approx 1.685,
\qquad
\eta=-\frac{L(\tilde\gamma_1)-aE(\tilde\gamma_1)}{\kappa}\approx 0.057>0,
\]
and condition (ii) is also satisfied, since
$E(\tilde\gamma_1)\approx 3.199>1$ and $A_{\eta,\kappa}\approx -5.338<0$.
Hence $\tilde\gamma_1$ is a flyby electromagnetic orbit with the same
turning value $\bar r$. Since
$E(\alpha)<0$, inequality \eqref{energy_extracted_particle} applies and
yields the energy gain
\[
E(\tilde\gamma_1)\approx 3.199>2.810\approx E(\gamma_1),
\]
realizing the Penrose-type process of Section~\ref{Penrose} in the
negative-$r$ region.

\smallskip
\noindent\emph{Step 4: the second flyby orbit and the second decay
	$\alpha\to\gamma_2+\tilde\alpha_1$.}
We choose $\gamma_2$ in the family of
Theorem~\ref{prescribed_common_turning_fyby} with
\[
\kappa=2,\ \text{i.e.}\ k_{\gamma_2}=-2,
\qquad
\eta\approx 0.077,\ \text{i.e.}\ D_{\eta,\kappa}\approx -0.154;
\]
then \eqref{prescribed turning energy} gives
$E(\gamma_2)\approx 3.3$ and $L(\gamma_2)\approx 9.746$, and
$A_{\eta,\kappa}\approx -2.975<0$, so that condition (ii) is satisfied. Thus,
by Theorem~\ref{prescribed_common_turning_fyby}, there exists another flyby
electromagnetic orbit $\gamma_2$ with turning value $\bar r$. Set
$m_{\gamma_2}=0.01$, hence $e_{\gamma_2}=-0.02$. Let the decay occur at a
point $w=\alpha(s)$, $s>0$, to be fixed in Step 5. By the conservation laws
and the normalization $g(\dot{\tilde\alpha}_1,\dot{\tilde\alpha}_1)=-1$,
\begin{align*}
	& m_{\tilde{\alpha}_1}\approx 0.152,\\
	&k_{\tilde{\alpha}_1}= \frac{k_{\alpha} m_{\alpha}-k_{\gamma_2}m_{\gamma_2}}{m_{\tilde{\alpha}_1}}\approx 22.597,\\
	&E(\tilde{\alpha}_1)=\frac{m_{\alpha}E(\alpha)-m_{\gamma_2}E(\gamma_2)}{m_{\tilde{\alpha}_1}}\approx -0.849,\\
	&L(\tilde{\alpha}_1)=\frac{m_{\alpha}L(\alpha)-m_{\gamma_2}L(\gamma_2)}{m_{\tilde{\alpha}_1}}\approx -17.103,\\
	&e_{\tilde{\alpha}_1}=k_{\tilde{\alpha}_1} m_{\tilde{\alpha}_1}\approx 3.429,
\end{align*}
with $g(\dot{\tilde{\alpha}}_1,V)\approx -94.848<0$ and
$m_\alpha>m_{\gamma_2}+m_{\tilde\alpha_1}$, so the decay is kinematically
admissible; $\tilde\alpha_1$ is an $r$-bouncing orbit with radial range
approximately $(-2.521,-2.5)$.
    
	\smallskip
	\noindent\emph{Step 5: the closed broken electromagnetic orbit.}
	The orbits $\gamma_1$ and $\gamma_2$ are equatorial flyby
	electromagnetic orbits with the same radial turning value $\bar r=-2.5$.
	By Theorem~\ref{closed_broken_orbit_blockIII} there exists a sequence
	$r_n\to -\infty$ satisfying \eqref{phasecondition} with
	$s_n=T_1(r_n)+T_2(r_n)>0$ for all $n$ large enough. Choosing the second
	decay at the point $w_n=\alpha(s_n)$, the concatenation
	\[
	\gamma_1\big|_{[\tau_n,0]}\cup\alpha\big|_{[0,s_n]}\cup\gamma_2\big|_{[0,\sigma_n]}
	\]
	is a closed broken electromagnetic orbit contained in block~$\mathrm{III}$.
\end{example}

\section{Conclusions}
The results of this paper show that the time-machine region of Kerr--Newman spacetime is not only a locus of causality violation at the level of the metric, but also supports distinguished closed timelike trajectories governed by the Lorentz force equation. This is the main difference with the neutral Kerr picture. The simplest manifestation of this phenomenon is provided by the closed electromagnetic orbits tangent to $\partial_\phi$. Their existence is rigid: in the positive-$r$ part of the equatorial time-machine region, fixing the radius fixes the charge-to-mass ratio, and conversely fixing an admissible charge-to-mass ratio fixes the radius. These orbits are not introduced by an external force added by hand; they are sustained by the Kerr--Newman electromagnetic field itself, provided that the black hole charge has opposite sign to that of the particle. 

The main  point of the paper is that closedness can persist in a weaker, but physically natural, sense when particle decays are allowed.  A decay process can be described as a concatenation of  Lorentz-force trajectories with different charge-to-mass ratios, provided the conservation laws at the decay vertices are satisfied.
We have shown that  different equatorial flyby
electromagnetic orbits can be chosen with the same radial turning value inside
the time-machine region $\mathfrak T$, both in the positive-$r$ and negative-$r$ regions. At such a turning point a decay can
produce an equatorial circular electromagnetic orbit $\alpha$ with the same radius. The
subsequent decay of $\alpha$ can  then produce a second equatorial flyby branch
$\gamma_2$, and  condition \eqref{phasecondition} in
Theorem~\ref{closed_broken_orbit_blockIII} ensures that, for a suitable  parameter value  on $\alpha$, the concatenation $\gamma_1\cup\alpha\cup\gamma_2$
is a closed  broken electromagnetic orbit contained in
block $\mathrm{III}$.

This causality-violating process is separated from the
Penrose-type interpretation of the first decay
$\gamma_1\to\alpha+\tilde\gamma_1$. In this case the spherical orbit  $\alpha$ can have negative conserved energy even though $\partial_t$ is still timelike and future-pointing at 
the decay point; the negative contribution comes from the electromagnetic
coupling in the conserved energy, rather than from an ergoregion. Conservation
of energy then forces the other daughter particle
$\tilde\gamma_1$ to carry more energy than the parent. Since the decay takes
place below the Cauchy horizon, however, this energy extraction does not involve block I:
the electromagnetic orbit $\tilde\gamma_1$ can be flyby, and so   it reaches a different asymptotically
flat end of the maximal analytic extension, provided the radial turning value is positive while, for
a negative radial turning value, the extracted energy is carried to the
asymptotically flat end $r\to-\infty$ of block III itself.

The analysis leaves some natural questions open. The first is whether smooth closed electromagnetic orbits exist  beyond the family tangent to $\partial_\phi$ in the time-machine region.  A second question concerns the non-equatorial dynamics.  
Finally, our construction is carried out in the test-particle approximation and does not address backreaction or self-force effects. The equatorial circular
electromagnetic orbits used in the decay process indicate where this approximation reaches its natural limit. Indeed, for the family of solutions
$\gamma(s)=(t_0-s,\bar r,\pi/2,\phi_0+\omega s)$
described in Theorem~\ref{proposition alpha curve}, the proper acceleration is
\begin{align*}
	a_c(\gamma)
	&:=
	\sqrt{
		g\bigl(\nabla_{\dot\gamma}\dot\gamma,
		\nabla_{\dot\gamma}\dot\gamma\bigr)
	}
	=
	\sqrt{
		k^2 g_{\mu\nu}
		F^\mu{}_\rho F^\nu{}_\sigma
		\dot\gamma^\rho \dot\gamma^\sigma
	} \\
	&=
	\sqrt{
		k^2 g_{rr}
		\bigl(F^r{}_t\dot\gamma^t+F^r{}_\phi\dot\gamma^\phi\bigr)^2
	}
	=
	\frac{\sqrt{\Delta(\bar r)}}{|\bar r|^3}
	\,|k\bh|\,|1+a\omega|.
\end{align*}
From \eqref{omega middle curve} and \eqref{kbarr}, as $\bar r\to0$ one has
\[
|k\bh|\longrightarrow |\bh|,
\qquad
|1+a\omega|\sim \frac{|\bar r|}{|\bh|}.
\]
Since \(\Delta(\bar r)\to a^2+\bh^2\), it follows that
\[
a_c(\gamma)
\sim
\frac{\sqrt{a^2+\bh^2}}{|\bar r|^2}
\longrightarrow+\infty
\qquad\text{as }\bar r\to0.
\]
Thus the decays considered here are arranged precisely in a region where self-force, backreaction and
possibly quantum effects cannot be expected to remain negligible.  Whether an analogous causality-violating mechanism persists once these effects are included remains an open problem.

\section*{Acknowledgments}
\noindent E. Caponio is  partially supported   by  MUR under the Program ``Department of Excellence'' Legge 232/2016  (Grant No. CUP - D93C23000100001) and by ``INdAM - GNAMPA Project''  CUP E53C25002010001.\\
E. Caponio and S. Suhr gratefully acknowledge support from the Simons Center for Geometry and Physics, Stony Brook University and in particular the Program ``Contact Geometry, General Relativity and Thermodynamics'' during which some of the research for this paper was performed.\\
G. Sanzeni gratefully acknowledges partial support for this research from the Pacific Institute for the Mathematical Sciences (PIMS-20260623-PDF). G. Sanzeni thanks Stefan Nemirovski, Don Page and Eric Woolgar for helpful discussions.
\addtocontents{toc}{\setcounter{tocdepth}{0}}
\section*{Declarations}
\subsection*{Conflict of interest} The authors have no competing interests to declare that are relevant to the content of this article.
\subsection*{Data availability statement}
Data sharing is not applicable to this article as no new data were created or analyzed in this study.

\addtocontents{toc}{\setcounter{tocdepth}{1}}

\end{document}